\documentclass{mcom-l}
\usepackage{amsmath}
\usepackage{amsaddr}
\usepackage{algorithmic}
\usepackage{algorithm}
\usepackage{comment,todonotes}

\usepackage{ifxetex}
\ifxetex
\usepackage{mathspec}
\usepackage{fontspec}
\else
\usepackage[utf8]{inputenc}
\fi

\usepackage[style=numeric-comp,maxnames=99,url=false,isbn=false]{biblatex}
\usepackage[margin=1in]{geometry}
\usepackage{xcolor}
\usepackage{tikz}
\usetikzlibrary{%
  matrix,%
  calc,%
  arrows,%
  positioning%
}
\usepackage{tikz-cd}

\usepackage{bbold}
\usepackage{color}
\usepackage{url}

\usepackage{amsaddr}

\usepackage{todonotes}
\newtheorem{theorem}{Theorem}
\newtheorem{lemma}[theorem]{Lemma}
\newtheorem{definition}[theorem]{Definition}

\newcommand\img{\mathop{\mathrm{img}}}
\newcommand\coker{\mathop{\mathrm{coker}}}

\newcommand\into\hookrightarrow
\newcommand\onto\twoheadrightarrow

\newcommand\kk{\mathbb{k}}
\newcommand\X{\mathbb{X}}

\newcommand\dfn[1]{\textbf{#1}}

\begin{document}

\title{Persistence modules: Algebra and algorithms}

\author{Primoz Skraba}
\address{Jo\v zef Stefan Institute, Ljubljana, Slovenia}
\address{Artificial Intelligence Laboratory, Jo\v zef Stefan Institute, 
Jamova 39, 1000 Ljubljana, Slovenia}
\email{\url{primoz.skraba@ijs.si}}
\author{Mikael Vejdemo-Johansson}
\address{Corresponding author}
\address{Formerly: School of Computer Science, University of St Andrews, Scotland}
\address{Computer Vision and Active Perception Lab, KTH, Teknikringen 14, 100 44 Stockholm, Sweden}
\email{\url{mvj@kth.se}}


\begin{abstract}
  Persistent homology was shown by \textcite{cz2005} to be homology
  of graded chain complexes with coefficients in the graded ring
  $\kk[t]$. As such, the behavior of \emph{persistence modules} ---
  graded modules over $\kk[t]$ --- is an important part in the analysis
  and computation of persistent homology.

  In this paper we present a number of facts about persistence
  modules; ranging from the well-known but under-utilized to the
  reconstruction of techniques to work in a purely algebraic approach
  to persistent homology. In particular, the results we present give
  concrete algorithms to compute the persistent homology of a
  simplicial complex with torsion in the chain complex.
\end{abstract}

\maketitle

\vspace{-14pt}
\tableofcontents
\newpage

\section{Introduction}\label{sec:introduction}

The ideas of topological persistence
\autocite{edelsbrunner2000topological} and persistent homology
\autocite{cz2005} have had a fundamental impact on computational
geometry and the newly spawned field of applied topology. The
method cleverly modified tools from algebraic topology to make
them resistant to noise and applicable to input from
applications. It has found many uses, and produces globally
descriptive results often completely inaccessible with other
methods. Theoretical work in the field has served to successively
strengthen the weight and applicability of the methods. For
surveys or books on the subject, we recommend: \autocite
{zomorodian2005topology,carlsson2009topology,ghrist2008barcodes,eh-ct-09}.

Since the initial publication
\textcite{edelsbrunner2000topological}, persistent homology has
spawned many results in applied and computational topology. It
has been applied to manifold learning~\cite{NSW1,NSW2},
bioinformatics~\cite{bioinf,bioinf2}, computational
chemistry~\cite{chem}, medical~\cite{medical,med2}, trace
data~\cite{Walker}, among others (see also
\cite{image-other3,dgh-persistence,other1,other2,socg-pbsds-10}). See
~\cite{edelsbrunner2008persistent} for a complete survey of the development the field. 

The introduction of algebraic insights into the field has already
been productive: \textcite{cz2005} identified persistence modules
as graded modules over the polynomial ring $\kk[t]$ over a
coefficient field $\kk$, and used this insight to produce
advances in the computation of persistent homology. Even later,
\textcite{CdSM09} recognize that the tame representation theory
of $A_n$-quivers allow an analogous construction producing
zig-zag persistence.

Building on this backdrop, we introduce techniques from
computational commutative algebra into the study of persistent
homology -- maintaining explicit presentations of persistence
modules enables us to give purely algebraic approaches to
kernels, cokernels and images of morphisms between persistence
modules -- providing an alternative to the approach by
\textcite{cohen2009persistent}. In addition to this, it provides
algorithms for a variety of other algebraic constructions
including algorithms to compute persistent homology of chain
complexes with torsion present, and on that route an algebraic
approach to relative persistent homology and cohomology.

\subsection{Prior work}
\label{sec:prior-work}

At the core of the study of topological persistence lies two papers:
\textcite{edelsbrunner2000topological} who introduced the idea of
persistent homology in the first place, working exclusively with
coefficients in $\mathbb Z/2\mathbb Z$ and giving an ad hoc algorithm
for computing persistence barcodes.

This was followed by \textcite{cz2005}, who identified persistence
modules as, essentially, graded modules over the graded polynomial
ring $\kk[t]$, and described based on this identification how
algorithms follow where $\kk$ may be an arbitrary field.

The algebraic formalization approach, combined with results from
\textcite{gabriel1972unzerlegbare} produced fundamental results by
\textcite{CdSM09}, identifying quiver algebras as interesting
coefficient rings, producing a theory of zig-zag persistent homology
that allows the study of diagrams of topological spaces that no longer
form strict filtrations. As such, these form one way to approach the
question of how to handle torsion, or vanishing simplices, in a
persistent homology approach.

The usefulness of persistent homology relies fundamentally on
stability results -- for persistent homology generated from sublevel
filtrations of topological spaces, a small change in the filtration
function generates a quantifiably small change in the resulting
persistence modules. These types of results were introduced to the
field by \textcite{cohen2007stability} who require the underlying
space to be triangulable and the persistent homology modules to be
tame and degree-wise finite dimensional. Their results were
significantly improved by \textcite{chazal2009proximity}, who
formalize persistent homology as $\mathbb R$-indexed diagrams of
vector spaces, and define an \emph{interleaving distance} between such
diagrams to provide stability results that no longer require
continuous filtration functions, triangulable spaces, or tame
persistent homology modules. These stability results culminate in
recent work by \textcite{bubenik2012categorification}, who study the
category of $\mathbb  R$-indexed diagrams in arbitrary abelian
categories, and are able to prove a range of stability results on
sublevel filtrations using arbitrary filtration functions for
arbitrary functors from topological spaces to some abelian category.

Algebraic operations in persistent homology have also been considered
before. The kernel, cokernel and image constructions we introduce here
are algebraic reformulations of the techniques introduced in
\textcite{cohen2009persistent}. 

Torsion in chain complexes is tightly linked to relative homology, a
tool that has shown up numerous times in the
field. \textcite{cohen2009extending} provide a framework for
persistence in which by including relative homology groups of the
entire space relative to a superset filtration, all barcode intervals
are finite in length. This approach is refined by
\textcite{carlsson2009zigzag}, who are able to fit extended
persistence into a \emph{Mayer-Vietoris pyramid}, connecting various
constellations of relative persistent homologies into a large and
expressive system.
\textcite{bendich2007inferring}, \textcite{bendich2012local}, and \textcite{wang2011branching}
use persistent versions of local homology -- homology relative the
complement of a small neighborhood of a point -- to discover local
behavior close to that point.

\subsection{Our contributions}
\label{sec:our-contributions}

In this paper, we introduce techniques from computational
commutative algebra to the study of persistence modules. In
Section \ref{sec:persistence}, we review the correspondence
between persistence modules and graded $\kk[t]$-modules. Section
\ref{sec:ring_theory} reviews fundamental concepts from
computational commutative algebra. Of importance to the rest of
the paper is the adaptation of Smith normal form in Section
\ref{sec:graded-smith-normal}, which gives a graded version of
the Smith normal form and its computation that we have been
unable to find in the literature. Several of the results in
Section \ref{sec:ring_theory} are regularly discussed either for
simpler rings (fields, integers) or vastly more complex rings
(algebraic geometry) -- we discuss the special case of graded
$\mathbb k[t]$-modules and their applications to persistence. In
particular, the maximum possible free dimension in the category
of graded $\mathbb k[t]$-modules is discussed with its
implications for persistence modules. The goal of this paper is
twofold: to expand on the algebraic interpretation of persistence
and provide a uniform framework for algebraic constructions using
persistence modules

In Section \ref{sec:constructions}, we start elaborating on
standard techniques from computational algebra for representing
finitely presented modules, and draw from these to give specific
constructions for persistence modules. Section
\ref{sec:two_modules} discusses how a presentation map from
relations to generators captures the behavior of a persistence
module, and allows a graded Smith normal form computation to
recover a barcode from an arbitrary presentation. The sections
\ref{sec:direct-sum}, \ref{sec:image}, \ref{sec:cokernel},
\ref{sec:kernel}, \ref{sec:pullback}, and \ref{sec:pushout}, are
dedicated to purely algebraic constructions of kernels,
cokernels, images, pushouts, and pullbacks of persistence module
maps.

In Section~\ref{sec:applications}, we discuss applications of the
techniques and algorithms that we have described. In
Section~\ref{sec:pers-rel-homol}, we show how the nested module
presentation of a quotient module together with kernel and
cokernel algorithms allows us to compute persistent homology of a
chain complex with torsion, giving a purely algebraic approach to
relative homology.

In Section~\ref{sec:spectral-sequences}, we discuss the relations
between the work in this paper, and forthcoming work on computing
spectral sequences of persistence modules.

In Section~\ref{sec:unord-pers-comp} we demonstrate how the
fundamental viewpoint of graded modules over a graded ring provides us
a variation of the persistence algorithm that takes an \emph{unsorted}
stream of simplices as input, and changes the output on the fly should
new simplices indicate that the previous inferences were inaccurate.

\section{Persistent (Co-)Homology}\label{sec:persistence}

In persistent homology and cohomology, the basic object of study
is a filtered simplicial complex, with the filtration often
induced as a sublevel- or superlevel-filtration of some
(real-valued) function on a space. Since the filtration can be
considered a special kind of representation in simplicial
complexes of the total order of the function domain, where all
the endo-maps induced in the representation by the total order
represented are injective, it is clear that because of
functoriality, we can take point-wise homology (with coefficients
in $\kk$) on the representation and thus get a functor
\[
\text{SpCpx}^{\text{Domain}} \to \text{Mod}_\kk^{\text{Domain}}
\]
taking each filtered simplicial complex to a directed system
of $\kk$-modules, in which the map $H_*(X_\epsilon)\to
H_*(X_{\epsilon'})$ is induced by the inclusion $X_\epsilon\into
X_{\epsilon'}$ for $\epsilon<\epsilon'$.

For finite simplicial complexes, it is clear that at most
finitely many critical points $x_i$ can exist such that for each
$x_i$, and for all sufficiently small $\varepsilon$, the map
$X_{x_i-\varepsilon}\into X_{x_i+\varepsilon}$ are not
isomorphisms. Hence, we can change the entire representation from
the original range to a representation over a finite order
\[
-\infty < (x_0+x_1)/2 < (x_1+x_2)/2 < \dots < (x_{n-1}+x_n)/2 < \infty
\]

The same observation holds for a countable simplicial complex,
which corresponds to a countable order bounded by the top and
bottom $\pm\infty$.

Assuming that the order is left bounded, or in other words, that
for any point along the order there are only finitely many
populated ``compartments'' less than the current one, there are
clear isomorphisms between chain complexes for these
representations and graded modules over the ring
$\kk[t]$. Indeed, given left boundedness, we can write for the
center points stated above of the compartments $a_0=-\infty,
a_1=(x_0+x_1)/2, \dots$ enumerating the compartments as living
above $a_0, a_1, \dots$. We introduce a $\mathbb N$-grading on
the chain complex by making an element that lives over $a_k$ have
degree $k$. Finally, we introduce a $\kk[t]$-action by letting
\begin{enumerate}
\item $r\cdot c$ be given by the $\kk$-module structure on the chain complex $C_*X_{a_k}$ for $r\in\kk$ and $c\in C_*X_{a_k}$
\item $t\cdot c$ be given by the induced map $C_*X_{a_k}\to C_*X_{a_{k+1}}$
\end{enumerate}

This extends by $\kk[t]$-linearity to a $\kk[t]$-module structure for the entire representation $C_*X_*$.

At the core of the observations by \textcite{cz2005}, and also by \textcite{CdSM09} is that the computation of persistent homology is really the computation of homology within the category of modules over the appropriate quiver algebra -- $\kk[t]$ or even $\kk[t]/t^n$ if we have classical persistence, and $\kk Q$ for a quiver $Q$ of type $A_n$ for zig zag persistence. This observation underlies our approach to algebraic methods in persistent homology.

Indeed, once we can identify $C_*X_*$ with a $\kk[t]$-module by the methods above, the remainder of a persistent homology computation follows immediately. The persistence algorithm as described by \textcite{cz2005} is a Gaussian elimination algorithm as applied to graded $\kk[t]$-modules, and the entire computation follows from this recognition.

\section{Overview of Commutative Algebra}
\label{sec:commutative}

In this section, we will recall important and relevant facts from
commutative algebra; for some of these results we will be using the
internals of the proofs in later developments. The entirety of
Section~\ref{sec:ring_theory} can be skipped if you are already
comfortably familiar with algebra.

Commutative algebra is an area of study which deals with commutative
rings and their associated objects. With a field, the lessons
learned from commutative algebra reduce to classical linear
algebra -- and so, for much of what we will be dealing with, the
commutative algebra can well be considered to be a generalization
of familiar techniques from linear algebra into areas where not
all the strength of linear algebra is available. 

In particular, the construction of a basis for an arbitrary
finite dimensional vector space is a fundamental building block
in many algorithms. It, however, breaks down if the coefficients
do not lie in a field. The core step is where the leading
coefficient in some potential basis vector gets reduced to 1 by
dividing all coefficients with this leading coefficient. In
general commutative rings, there is no guarantee for having
multiplicative inverses, and therefore, this step can fail.

Instead, in commutative algebra, for a vector space like object
-- a module -- to have a basis is a very particular feature of
the module, and is enough to warrant a separate name for that
kind of module: a free module. As we shall see later on, a lot of
the power comes from phrasing our questions and our objects in
ways that retain free modules in the descriptions and solutions
as much as possible.

\subsection{Ring Theory}\label{sec:ring_theory}

A \dfn{ring} is some set of elements together with two binary
operations: addition (denoted by $+$) and multiplication (denoted
by $*$, $\cdot$ or juxtaposition). These are made to follow
reasonable axioms: addition should make the entire ring into an
abelian group. In particular, addition is commutative,
associative, and for each element $a$, it has an additive inverse
$-a$, as well as an additive identity element $0$ in the ring.

Multiplication also follows axioms: associativity, and
commutativity (hence \emph{commutative algebra}), and we shall
also always require a multiplicative identity element (denoted by
$1$).

Finally, the two operations are tied together by
\emph{distributivity}: $(a+b)\cdot c = a\cdot c + b\cdot c$.

The notable difference from the well-known definition of a field
is the lack a requirement of multiplicative inverses: indeed, if
each non-zero element $a$ has a multiplicative inverse $a^{-1}$
such that $a\cdot a^{-1}=1$, then the commutative ring is in fact
a field.

A ``vector space'' $M$ over a commutative ring $R$ is called a
\dfn{module} (or more accurately, a vector space is a module over
a field). More precisely, a module is an abelian group with a
specific binary operation $\cdot: R\times M\to M$ which obeys the
following axioms:
\begin{itemize}
\item $0\cdot m = 0$, $1\cdot m = 1$, 
\item $r\cdot (s\cdot m) =
(r\cdot s)\cdot m$; $\quad(r+s)\cdot m = r\cdot m + s\cdot m$, 
\item $r\cdot (m+n) = r\cdot m + r\cdot n$ for $r, s\in R$ and $m, n\in
M$.
\end{itemize}

We now state a few technical definitions followed by some informal comments. 

A \dfn{submodule} $N \leq M$ is a subgroup $N\subseteq M$ such that for any $r\in R$ and any $n\in N$, $r\cdot n\in N$. This is analogous to subspaces of vector spaces. The submodule is \dfn{trivial} if it only contains the $0$ element, and it is \dfn{proper} if it is a strict subset of $M$. Given a submodule $N\leq M$, the quotient group $M/N$ has the structure of an $R$-module by defining $r\cdot[m]=[r\cdot m]$. We call the resulting module the \dfn{quotient module} $M/N$. A \dfn{presentation} of a module $M$ is a presentation of $M$ as a quotient module $F/K$. The presentation is \dfn{free} if both $F$ and $K$ are free modules. The module $M$ is \dfn{finitely generated} if a presentation exists with $F$ free with a finite basis. The module $M$ is \dfn{finitely presented} if both $F$ and $K$ are free modules with finite bases.

A submodule $N < R$ of the ring itself is called an \dfn{ideal}.

A ring is \dfn{graded} over a monoid $(G,+)$ if the ring decomposes as a direct sum $R = \bigoplus_{g\in G} R_g$ where each $R_g$ is a subgroup of the additive group of $R$, and for elements $r\in R_g$ and $s\in R_h$, the product is in $R_{g+h}$. A module $M$ over a graded ring $R$ is graded if it decomposes into a direct sum of subgroups $M = \bigoplus_{g\in G} M_g$ such that the scalar product obeys the grading: $r\in R_g$ and $m\in M_h$ implies $r\cdot m\in M_{g+h}$. 

A ring is \dfn{noetherian} if any ascending chain of ideals $I_0 \subseteq I_1 \subseteq \dots$ eventually stabilizes. In other words, $R$ is noetherian if for every such chain of ideals there is some integer $N$ such that for all $n>N$, $I_n=I_{n+1}$. Very many rings we know are noetherian, in particular $\mathbb Z$, all fields, all principal ideal domains, and all polynomial rings over noetherian rings. Any ideal in a noetherian ring is finitely generated, and the partially ordered (under inclusion) set of any non-empty set of ideals of a noetherian ring has a maximal element.

Any ring has at least one proper ideal: the set $\{0\}$ is always an ideal, usually denoted by $0$ with some abuse of notation. If $0$ is the only proper ideal in a ring $R$, then that ring is a field. Indeed, suppose $r\neq 0$. Then, since $0$ is the only proper ideal, the principal ideal $(r)$ has to include the entire ring. In particular, $1\in (r)$, and thus there is some $s$ such that $rs=1$. Hence, $r$ has a multiplicative inverse.

\begin{theorem}[Decomposition of modules over PIDs]
  \label{thm:decomposition-module-PID}
  Suppose $R$ is a principal ideal domain. Then for every finitely generated module $M$ over $R$, there is an integer $n$ and a sequence of elements $r_1,\dots,r_m$ such that
  \[
  M = R^n \oplus \left( \bigoplus_{i=1}^m M/(r_i) \right)
  \]
\end{theorem}

The proof of this is standard, and relies on the observation that matrices over a principal ideal domain admit a Smith normal form.

\subsubsection{Euclidean Domains}\label{sec:euclidean_domains}

The following results and their proofs are well-known facts of commutative and computational algebra. We shall, however, reproduce some of them in order to be able to refer to their particular structure in this paper.

\begin{definition}
  Suppose $R$ is a ring. A \emph{Euclidean function} on $R$ is a function
  $R\to\mathbb N$ such that if $a, b\in R$ and $b\neq 0$, then there
  are $q, r\in R$ such that $a = b\cdot q+r$ and either $r=0$ or $\deg
  r < \deg b$. 

  An integral domain $R$ that supports at least one Euclidean function
  is called a \emph{Euclidean domain}.
\end{definition}

\begin{theorem}
  In a Euclidean domain, any two elements have a GCD.
\end{theorem}
\begin{proof}
  The extended Euclidean algorithm relies precisely on the existence
  of a Euclidean function, and generates, for elements $r, s$ 
  elements $a,b,g$ such that $g=\gcd(r,s)$ and $a\cdot r+b\cdot s=g$.
\end{proof}

\begin{theorem}
  All Euclidean domains are PIDs.
\end{theorem}
\begin{proof}
  Suppose $I\subseteq R$ is an ideal in a Euclidean domain. By
  well-order of $\mathbb N$, there is some minimal value in the set
  $\{\deg r | r\in I\}$, say $n$. Again, by the well-ordering of
  $\mathbb N$, this value is attained -- so there is some $b\in I$
  such that $\deg b = n$.

  We claim that $I$ is generated by $b$. Indeed, suppose that $a\in I$
  is some other element of $I$. Since $\deg a\geq\deg b$ by the
  minimality of $\deg b$, there is some $q, r$ such that $a = b\cdot q
  + r$ and either $r=0$ or $\deg r < \deg b$. Since $r=a-b\cdot q$,
  and both $a, b\in I$, it follows that $r\in I$.

  Hence, if $r\neq 0$, then $\deg r < \deg b$ contradicting the
  minimality of degree in our choice of $b$. Hence, $r=0$ and thus $a
  = b\cdot q$, so $a\in\langle b\rangle$.

  The result follows immediately.
\end{proof}

\subsubsection{Connection with Persistence}
\label{sec:conn_persistence}

\textcite{cz2005} identify the \emph{persistence modules} 
with graded modules over $\kk[t]$. We remind the reader of the
intuition of this representation. Recall that persistent homology
computes the homology over a nested sequence of spaces called a filtration:
\begin{eqnarray*}
\emptyset  = \X_0 \subseteq\X_1 \subseteq \X_2 \subseteq\ldots \subseteq\X_{N-1} \subseteq \X_{N} =\X 
\end{eqnarray*}
The grading comes from viewing each of the spaces in the sequence
as a slice and stacking the slices up as in Figure~\ref{fig:zomorodian-example-graded}.

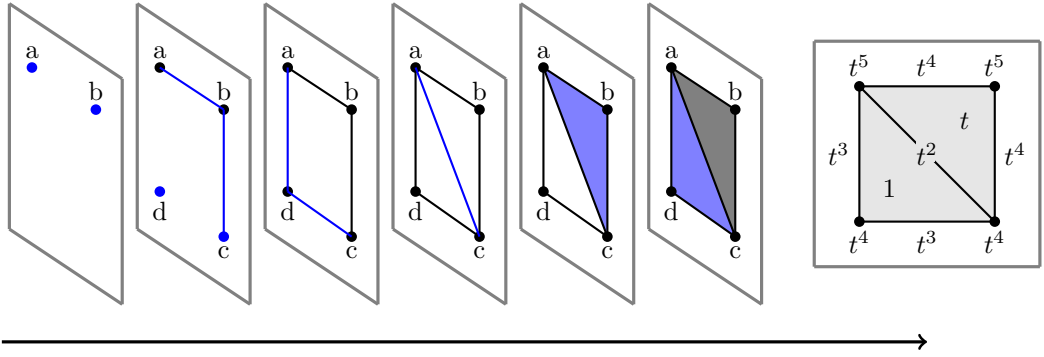
\begin{figure}
  \centering
  \begin{tikzpicture}[thick]
    \begin{scope}[xshift=-1.7cm]
      \node [coordinate,label=above:a] (a) at (-0.2,1.65) {};
      \node [coordinate,label=above:b] (b) at (0.65,1.09) {};
      \fill [blue] (a) circle (2pt);
      \fill [blue] (b) circle (2pt);
      \draw [very thick,gray] (-0.5,-0.5) -- (1,-1.5);
      \draw [very thick,gray] (1,-1.5) -- (1,1.5);
      \draw [very thick,gray] (1,1.5) -- (-0.5,2.5);
      \draw [very thick,gray] (-0.5,-0.5) -- (-0.5,2.5);
    \end{scope}
    \begin{scope}[xshift=0cm]
      \node [coordinate,label=above:a] (a) at (-0.2,1.65) {};
      \node [coordinate,label=above:b] (b) at (0.65,1.09) {};
      \node [coordinate,label=below:d] (d) at (-0.2,0) {};
      \node [coordinate,label=below:c] (c) at (0.65,-0.6) {};
      \fill (a) circle (2pt);
      \fill (b) circle (2pt);
      \fill [blue] (c) circle (2pt);
      \fill [blue] (d) circle (2pt);   
      \draw [blue] (a) -- (b);
      \draw [blue] (c) -- (b);    
      \draw [very thick,gray] (-0.5,-0.5) -- (1,-1.5);
      \draw [very thick,gray] (1,-1.5) -- (1,1.5);
      \draw [very thick,gray] (1,1.5) -- (-0.5,2.5);
      \draw [very thick,gray] (-0.5,-0.5) -- (-0.5,2.5);
    \end{scope}
    \begin{scope}[xshift=1.7cm]
      \node [coordinate,label=above:a] (a) at (-0.2,1.65) {};
      \node [coordinate,label=above:b] (b) at (0.65,1.09) {};
      \node [coordinate,label=below:d] (d) at (-0.2,0) {};
      \node [coordinate,label=below:c] (c) at (0.65,-0.6) {};
      \fill (a) circle (2pt);
      \fill (b) circle (2pt);
      \fill (c) circle (2pt);
      \fill (d) circle (2pt);   
      \draw (a) -- (b);
      \draw (c) -- (b);    
      \draw [blue] (a) -- (d);
      \draw [blue] (c) -- (d);    
      \draw [very thick,gray] (-0.5,-0.5) -- (1,-1.5);
      \draw [very thick,gray] (1,-1.5) -- (1,1.5);
      \draw [very thick,gray] (1,1.5) -- (-0.5,2.5);
      \draw [very thick,gray] (-0.5,-0.5) -- (-0.5,2.5);
    \end{scope}
    \begin{scope}[xshift=3.4cm]
      \node [coordinate,label=above:a] (a) at (-0.2,1.65) {};
      \node [coordinate,label=above:b] (b) at (0.65,1.09) {};
      \node [coordinate,label=below:d] (d) at (-0.2,0) {};
      \node [coordinate,label=below:c] (c) at (0.65,-0.6) {};
      \fill (a) circle (2pt);
      \fill (b) circle (2pt);
      \fill (c) circle (2pt);
      \fill (d) circle (2pt);   
      \draw (a) -- (b);
      \draw (c) -- (b);     
      \draw  (a) -- (d);
      \draw  (c) -- (d);    
      \draw [blue] (a) -- (c);
      \draw [very thick,gray] (-0.5,-0.5) -- (1,-1.5);
      \draw [very thick,gray] (1,-1.5) -- (1,1.5);
      \draw [very thick,gray] (1,1.5) -- (-0.5,2.5);
      \draw [very thick,gray] (-0.5,-0.5) -- (-0.5,2.5);
    \end{scope}
    \begin{scope}[xshift=5.1cm]
      \node [coordinate,label=above:a] (a) at (-0.2,1.65) {};
      \node [coordinate,label=above:b] (b) at (0.65,1.09) {};
      \node [coordinate,label=below:d] (d) at (-0.2,0) {};
      \node [coordinate,label=below:c] (c) at (0.65,-0.6) {};
      \fill [blue!50] (a) -- (b) -- (c) -- cycle;
      \fill (a) circle (2pt);
      \fill (b) circle (2pt);
      \fill (c) circle (2pt);
      \fill (d) circle (2pt);   
      \draw (a) -- (b);
      \draw (c) -- (b);    
      \draw  (a) -- (d);
      \draw  (c) -- (d);    
      \draw  (a) -- (c); 
      \draw [very thick,gray] (-0.5,-0.5) -- (1,-1.5);
      \draw [very thick,gray] (1,-1.5) -- (1,1.5);
      \draw [very thick,gray] (1,1.5) -- (-0.5,2.5);
      \draw [very thick,gray] (-0.5,-0.5) -- (-0.5,2.5);
    \end{scope}
    \begin{scope}[xshift=6.8cm]
      \node [coordinate,label=above:a] (a) at (-0.2,1.65) {};
      \node [coordinate,label=above:b] (b) at (0.65,1.09) {};
      \node [coordinate,label=below:d] (d) at (-0.2,0) {};
      \node [coordinate,label=below:c] (c) at (0.65,-0.6) {};
      \fill [black!50] (a) -- (b) -- (c) -- cycle;
      \fill [blue!50] (a) -- (c) -- (d) -- cycle;
      \fill (a) circle (2pt);
      \fill (b) circle (2pt);
      \fill (c) circle (2pt);
      \fill (d) circle (2pt);   
      \draw (a) -- (b);
      \draw (c) -- (b);    
      \draw  (a) -- (d);
      \draw  (c) -- (d);    
      \draw  (a) -- (c); 
      \draw [very thick,gray] (-0.5,-0.5) -- (1,-1.5);
      \draw [very thick,gray] (1,-1.5) -- (1,1.5);
      \draw [very thick,gray] (1,1.5) -- (-0.5,2.5);
      \draw [very thick,gray] (-0.5,-0.5) -- (-0.5,2.5);
    \end{scope}
  \begin{scope}[xshift=9cm]
      \node [coordinate,label=above:$t^5$] (a) at (0.1,1.4) {};
      \node [coordinate,label=above:$t^5$] (b) at (1.9,1.4) {};
      \node [coordinate,label=below:$t^4$] (d) at (0.1,-0.4) {};
      \node [coordinate,label=below:$t^4$] (c) at (1.9,-0.4) {};
      \fill [lightgray!40] (a) -- (b) -- (c) -- cycle;
      \fill [lightgray!40] (a) -- (c) -- (d) -- cycle;
      \fill (a) circle (2pt);
      \fill (b) circle (2pt);
      \fill (c) circle (2pt);
      \fill (d) circle (2pt);   
      \draw (a) -- (b) node [midway, above] {$t^4$} ;
      \draw (c) -- (b) node [midway, right] {$t^4$};
      \draw  (a) -- (d) node [midway, left] {$t^3$};
      \draw  (c) -- (d) node [midway, below] {$t^3$};
      \draw  (a) -- (c) node [midway] {$$};
      \node [coordinate,label=$t$] (t1) at (1.5,0.7){};
      \node [coordinate,label=$1$] (t1) at (0.5,-0.2){};
      \draw [very thick,gray] (-0.5,-1) -- (2.5,-1);
      \draw [very thick,gray] (2.5,-1) -- (2.5,2);
      \draw [very thick,gray] (2.5,2) -- (-0.5,2);
      \draw [very thick,gray] (-0.5,-1) -- (-0.5,2);
      \fill [lightgray!40] (0.85,0.4) rectangle (1.1,0.9);          
      \path  (a) -- (c) node [midway] {$t^2$};
    \end{scope}
\draw[very thick,->]  (-2.3,-2) -- (10,-2);
  \end{tikzpicture}
  
\caption{Persistent homology example by \textcite{edelsbrunner2000topological}. The arrow denotes the direction of time in the filtration with the graded complex shown at the far right. }
  \label{fig:zomorodian-example-graded}
\end{figure}

This entire structure is then encoded in the final complex (top
slice), with each simplex annotated by a number indicating how
many slices ago it first appeared. This is precisely the
information encoded in the $\kk[t]$ module structure.

Furthermore, the free and torsion parts of the
decomposition in Theorem \ref{thm:decomposition-module-PID}
correspond to infinite and finite intervals in the persistent
barcodes respectively. With these results in hand, we can state the
following results. 

\begin{theorem}
  The graded ring $\kk[t]$ is a Euclidean domain.
\end{theorem}
\begin{proof}
  As a Euclidean function, take $\deg\sum a_it^i = \max\{i | a_i\neq 0\}$. We notice that this coincides with the classical notion of degree of a polynomial, and with the grading of $\kk[t]$.
\end{proof}

\begin{theorem}
  Any finitely generated submodule $M$ of a free module $F$ is itself
  free.
\end{theorem}
\begin{proof}
  Suppose $M\subseteq F$ is generated by elements
  $m_1,\dots,m_k$. Suppose that $M$ is not free. Then there is some
  $\kk[t]$-linear relation 
  \[
  \sum_i a_i t^{e_i} m_i = 0
  \]
  between the generators. If $e_i>0$ whenever $a_i\neq 0$, then the
  entire expression is a multiple of $t$. Since $F$ is free, if
  $t\cdot m=0$ then $m=0$ for all elements $m\in F$. In particular,
  this holds for $t\cdot\sum_i a_it^{e_i-1}m_i$. Thus, we reduce the
  exponents in this expression by some integer $d$ until for some $i$,
  $a_i\neq 0$ and $e_i=0$. Thus, this relation is equivalent to the
  relation
  \[
  a_im_i = -\sum_{j\neq i} a_jt^{e_j-d} m_j
  \]
  
  But this implies that $m_i$ is in the span of $m_1,\dots,\hat
  m_i,\dots,m_k$, where we retain the convention that the $\hat m_i$ means that we exclude
  $m_i$ from the list. Thus, the same module $M$ is also generated by
  the remaining generators.

  Since there are finitely many generators, this process of
  elimination eventually terminates. When it terminates, it must be
  because there are no more non-zero $\kk[t]$-linear relations to be
  found.

  Hence, if $M$ is a finitely generated submodule of $F$, it must be
  free.
\end{proof}

\begin{theorem}
  Any finite presentation $K\to G\to M\to 0$ of a graded module $M$
  over $\kk[t]$ is a free presentation $0\to K\to G\to M\to 0$.
\end{theorem}
\begin{proof}
  Suppose $G$ is a free module of generators of $M$, and
  $G\xrightarrow{p}M\to 0$ is the map of basis element to
  generator. $K$, the module of relations of the presentation is by
  force the free module of generators of $\ker p$. 

  Claim: The corresponding map $K\xrightarrow{i}G$ is an injective
  map.

  Indeed, suppose $i$ is not injective as a degree 0 map of graded
  $\kk[t]$-modules. Then there is some $k\neq 0$ such that $ik=0$. If
  $k$ is not homogeneous, then $k=\sum_d k_d$, and there is some
  non-zero homogeneous $k_d$ such that $ik_d=0$.

  Consider the map restricted to the degree $d$ part of all modules
  involved: $K_d\xrightarrow{i_d}G_d\xrightarrow{p_d}M_d\to 0$. Since
  everything is graded, this restriction of a presentation of $M$ is a
  finite presentation of $M_d$ as a $\kk$-module, or in other words as
  a vector space. Since all vector spaces are free modules, it follows
  that $\ker p_d$ is a free subspace of $G_d$, and thus WLOG, $\ker
  p_d = K_d$ with $i_d$ the inclusion map of $\ker p_d\subseteq
  G_d$. Therefore, $i_d$ is an injective map, and thus has a trivial
  kernel. It follows that $i_dk_d=0$ implies $k_d=0$. But this
  contradicts the assumption above. Hence $i$ is injective, and the
  result follows.
\end{proof}

It is worthwhile to demonstrate exactly how  this representation works. Chains in a graded $\kk[t]$-module of chains will have an expression that depends both on the birth times of their basis elements and on the point in time that the chain itself is considered. Consider the chain $ab+ac+cd$ in Figure~\ref{fig:zomorodian-example-graded}, existing at the 5th time step of the diagram. To find a $\kk[t]$-linear description of this chain, we need to promote each of the three basis elements, $ab$, $ac$ and $cd$, into the 5th time step. Since the edge $ab$ shows up in the second time step, this takes a coefficient of $t^3$. Similarly, $ac$ has a coefficient of $t$ and $cd$ a coefficient of $t^2$. All in all, the chain is $t\cdot ac+t^2\cdot cd+t^3\cdot ab$.

This is also relevant for combining elements of graded $\kk[t]$-modules: each (homogeneous) element has a degree it lives in --- and to sum different elements, ensuring a homogeneous result, the elements need to live in the same degree; the cycle basis from the situation in Figure~\ref{fig:zomorodian-example-graded} has as natural choices of basis elements the two cycles $z=ad-cd-t\cdot ab-t\cdot bc$ and $w=ac-t^2\cdot ab-t^2\cdot bc$. To express the other small cycle at time step 4, we will need to form a linear combination of the two cycles at hand \emph{at the time step we are interested in}. Hence, this other triangular cycle has the expression $w-t\cdot z=ac-t\cdot ad-t\cdot cd$.

\subsubsection{Graded Smith normal form}
\label{sec:graded-smith-normal}

The Smith normal form of a function $M\xrightarrow{f}N$ will compute a
simultaneous basis change for $M$ and $N$ such that in the new basis,
each basis element for $M$ pairs with a basis element in $N$; the map
takes any basis element to a scalar multiple of its paired element.

While developed for matrices over $\mathbb Z$ or $\mathbb N$, the
constructions all hold over any PID -- the Smith normal form is one of
the standard ways to prove the decomposition
Theorem~\ref{thm:decomposition-module-PID}. We shall describe how it
is usually computed, with special attention to the shortcuts and
differences that the case of graded modules over $\kk[t]$
introduce. We have been unable to find such an adaptation in the literature.

The usual algorithm allows the user to permute both rows and columns
of the matrix under treatment; since the grading is important to us,
we shall instead mark rows and columns as finished -- rather than to
permute them. 

We represent a map $M\xrightarrow{f}N$ by a matrix $F=(f_{ij})$ such
that $f(m_i) = \sum_j f_{ij}n_j$, where $m_i$ and $n_j$ are basis
elements in $M$ and $N$ respectively. 

A basis change that replaces $n_j$ by $n_j-rn_{j'}$ has the effect of
adding $r$ times the $j$th row to the $j'$th row. Similarly, a basis
change replacing $m_i$ by $m_i-rm_{i'}$ adds $r$ times the $i$th
column to the $i'$th column. We will be performing these operations
from left to right, from bottom to top, in a matrix that has been
pre-sorted to keep both bases of $M$ and $N$ in ascending degree
order.

The resulting algorithm is essentially forced by the requirement that
we can only add basis multiples of compatible orders: if
$|m_i|<|m_{i'}|$, we cannot influence $m_i$ using $m_{i'}$ -- we can
only influence $m_{i'}$ using $m_i$. Hence, the entries in the matrix
can only flow upwards in their influence; never downwards. 

\noindent\textbf{Algorithm:}
\begin{enumerate}
\item\label{algstep:loopstart} While there are untreated rows or
  columns: \\
  pick the lowest degree entry $f_{ij}$ of the
  lowest degree block of untreated columns.
\item By appropriate basis changes in $M$, clear out all entries of
  column $i$. Entries are cleared out upwards in the matrix.
\item By appropriate basis changes in $N$, clear out all entries of
  row $j$. Entries are cleared out leftwards in the matrix.
\item Mark row and column as finished.
\item Goto step \ref{algstep:loopstart}
\end{enumerate}


For extra clarity, we include two examples here. Consider the sequence
in Figure~\ref{fig:zomorodian-example-graded}. We can represent the chain
complex by a graded $\kk[t]$-module with chain basis $a, b, c, d, ab,
bc, ad, cd, ac, abc, acd$ and boundaries $tb-ta, c-tb, td-tc, td-t^2a,
t^2c-t^3a, t^3ab+t^3bc-tac, t^3cd-t^3ad+t^2ac$. Computing a kernel of
the boundary map, we get a basis for the cycle module given by:
\begin{align*}
  z_1 &= a & z_3 &= c & z_5&= tab+tbc+cd-ad\\
  z_2 &= b & z_4 &= d & z_6&= t^2ab+t^2bc-ac\\
\end{align*}

Writing the boundaries on this basis, we get a boundary module with a
basis given by $r_1, r_2, r_3, r_4, r_5, r_6, r_7$, and images in the
cycle basis given by 
\begin{align*}
  r_1 &\mapsto tz_2-tz_1 & 
  r_3 &\mapsto tz_4-tz_3 & 
  r_5 &\mapsto t^2z_3-t^3z_1 &
  r_7 &\mapsto t^3z_5 - t^2z_6 \\
  r_2 & \mapsto z_3-tz_2 & 
  r_4 &\mapsto tz_4-t^2z_1 & 
  r_6 &\mapsto tz_6 \\
\end{align*}

This is the presentation map for the finitely generated
$\kk[t]$-module that represents the persistent homology of the
sequence of spaces in Figure~\ref{fig:zomorodian-example-graded}, and by
computing a Smith normal form, we extract a barcode from the
example. We shall trace the corresponding matrix and its
modifications in Figure~\ref{fig:persistence-gsnf}. In each step, we mark the chosen entry $f_{ij}$ for the
next clearing out.
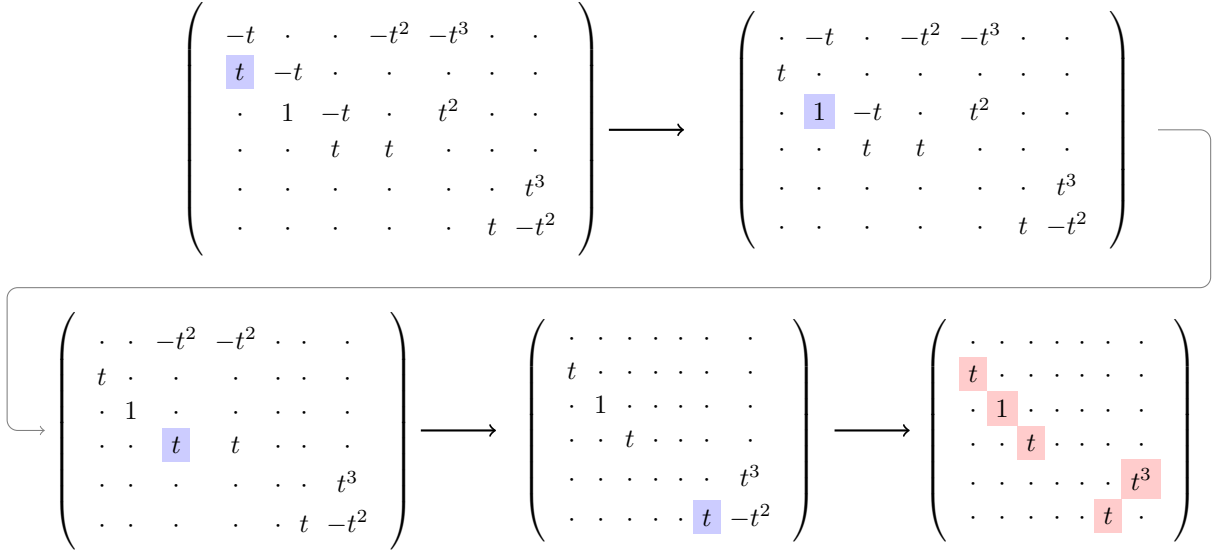
\begin{figure}
\begin{tikzpicture}
\begin{scope}[shift={(0.1,0)}]
  \matrix (M1) 
  [matrix of math nodes,
  left delimiter=(,
  right delimiter=)] 
  {
    -t & \cdot & \cdot & -t^2 & -t^3 & \cdot & \cdot \\
    \node[fill=blue!20] {t}; & -t &\cdot &\cdot &\cdot&\cdot&\cdot \\
    \cdot& 1 & -t &\cdot & t^2 &\cdot &\cdot \\
    \cdot&\cdot & t & t &\cdot&\cdot&\cdot  \\
    \cdot&\cdot &\cdot &\cdot &\cdot &\cdot & t^3 \\
    \cdot&\cdot &\cdot &\cdot &\cdot & t & -t^2 \\
  };
\end{scope}
\begin{scope}[shift={(7.3,0)}]
  \matrix (M1) 
  [matrix of math nodes,
  left delimiter=(,
  right delimiter=)] 
  {
    \cdot & -t & \cdot & -t^2 & -t^3 & \cdot & \cdot \\
    t & \cdot &\cdot &\cdot &\cdot&\cdot&\cdot \\
    \cdot& \node[fill=blue!20] {1}; & -t &\cdot & t^2 &\cdot &\cdot \\
    \cdot&\cdot & t & t &\cdot&\cdot&\cdot  \\
    \cdot&\cdot &\cdot &\cdot &\cdot &\cdot & t^3 \\
    \cdot&\cdot &\cdot &\cdot &\cdot & t & -t^2 \\
  };
\end{scope}
\begin{scope}[shift={(-2,-4)}]
  \matrix (M1) 
  [matrix of math nodes,
  left delimiter=(,
  right delimiter=)] 
  {
    \cdot & \cdot & -t^2 & -t^2 & \cdot & \cdot & \cdot \\
    t & \cdot &\cdot &\cdot &\cdot&\cdot&\cdot \\
    \cdot& 1 & \cdot &\cdot & \cdot &\cdot &\cdot \\
    \cdot&\cdot & \node[fill=blue!20] {t}; & t &\cdot&\cdot&\cdot  \\
    \cdot&\cdot &\cdot &\cdot &\cdot &\cdot & t^3 \\
    \cdot&\cdot &\cdot &\cdot &\cdot & t & -t^2 \\
  };
\end{scope}
\begin{scope}[shift={(3.8,-4)}]
  \matrix (M1) 
  [matrix of math nodes,
  left delimiter=(,
  right delimiter=)] 
  {
    \cdot & \cdot & \cdot & \cdot & \cdot & \cdot & \cdot \\
    t & \cdot &\cdot &\cdot &\cdot&\cdot&\cdot \\
    \cdot& 1 & \cdot &\cdot & \cdot &\cdot &\cdot \\
    \cdot&\cdot & t & \cdot &\cdot&\cdot&\cdot  \\
    \cdot&\cdot &\cdot &\cdot &\cdot &\cdot & t^3 \\
    \cdot&\cdot &\cdot &\cdot &\cdot & \node[fill=blue!20] {t}; & -t^2 \\
  };
\end{scope}
\begin{scope}[shift={(9,-4)}]
  \matrix (M1) 
  [matrix of math nodes,
  left delimiter=(,
  right delimiter=)] 
  {
    \cdot & \cdot & \cdot & \cdot & \cdot & \cdot & \cdot \\
     \node[fill=red!20] {t}; & \cdot &\cdot &\cdot &\cdot&\cdot&\cdot \\
    \cdot& \node[fill=red!20] {1}; & \cdot &\cdot & \cdot &\cdot &\cdot \\
    \cdot&\cdot & \node[fill=red!20] {t}; & \cdot &\cdot&\cdot&\cdot  \\
    \cdot&\cdot &\cdot &\cdot &\cdot &\cdot & \node[fill=red!20] {t^3}; \\
    \cdot&\cdot &\cdot &\cdot &\cdot & \node[fill=red!20] {t}; & \cdot \\
  };
\end{scope}
\draw[thick,->] (3,0) -- (4,0);
\draw[thick,->] (0.5,-4) -- (1.5,-4);
\draw[thick,->] (6,-4) -- (7,-4);

\draw[->,gray,rounded corners] (10.3,0) -- (11,0) -- (11,-2.1) -- (-5,-2.1) -- (-5,-4) -- (-4.5,-4);
\end{tikzpicture}
    \caption{First example: graded Smith normal form reduction for a persistence chain complex.}
  \label{fig:persistence-gsnf}
\end{figure}

From this endstate, we can easily read off the barcode -- especially
knowing that all our operations have maintained the degrees of rows
and columns; we have a new basis $z_1',z_2', z_3', z_4', z_5', z_6'$ of
the cycle module, and in this basis, the boundaries are given by
$tz_2', z_3', tz_4', t^3z_5', tz_6'$. The resulting barcode has
entries $(1,\infty), (1,2), (2,2), (2,3)$ in dimension 0 and entries
$(3,6), (4,5)$ in dimension 1.

In above example, we can read off the map $f$ directly since the
filtration and simplicial complex is small. The map here is the
map from the space of boundaries into the space of cycles
$B\rightarrow Z$. By the property of the boundary operator,
$\partial\cdot\partial = 0$, $B\subseteq Z$. Hence, we can
express each element of the boundary basis as a linear
combination of elements in the cycle basis. This gives us an
equation for each boundary basis element:
\begin{equation*}
b_i =  \sum_j f_{ij} z_j
\end{equation*}
we can find $f_{ij}$ by reducing $b_i$ with respect to $Z$,
giving us an explicit representation of the map.

Our second example is more abstract illustrating how the
algorithm works without an explicit chain space.

Suppose a persistence module has one presentation given by five
generators $x,y,z,u,v$ in degrees 1, 1, 2, 3, 3,  and four relations
$z+tx+ty, u+t^2x+t^2y, tv+t^2z+t^3y, tu+t^2z+t^3y$. We illustrate the
computation of its graded Smith normal form in
Figure~\ref{fig:module-gsnf}. From the end-state of the computation,
we can read off a new presentation with the generators $x,y',z',u',v'$
and relations $z', u', tv, t^3y'$, where $y'=y+2x, z'=z+ty+tx, u'=u+t^2y+t^2x,
v'=v-t^3x$. 

\begin{figure}
\hspace*{-2em}
 \begin{tikzpicture}
\begin{scope}[shift={(-1.5,0)}]
  \matrix [matrix of math nodes, left delimiter=(, right delimiter=)]
  (m) {
    t & t^2 & 0 & 0 \\
    t & t^2 & t^3 & t^3 \\
    \node [fill=blue!20] {1}; & 0 & t^2 & t^2 \\
    0 & 1 & 0 & t \\
    0 & 0 & t & 0 \\
  };
  \node [above=5pt of m-1-1,rotate=90,anchor=west] (top1) {$r_1$};
  \node [above=5pt of m-1-2,rotate=90,anchor=west] (top2) {$r_2$};
  \node [above=5pt of m-1-3,rotate=90,anchor=west] (top2) {$r_3$};
  \node [above=5pt of m-1-4,rotate=90,anchor=west] (top2) {$r_4$};
  \node [left=12pt of m-1-1,anchor=east] (left1) {$x$};
  \node [left=12pt of m-2-1,anchor=east] (left2) {$y$};
  \node [left=12pt of m-3-1,anchor=east] (left3) {$z$};
  \node [left=12pt of m-4-1,anchor=east] (left4) {$u$};
  \node [left=12pt of m-5-1,anchor=east] (left5) {$v$};
\end{scope}
\begin{scope}[shift={(2.7,0)}]
  \matrix [matrix of math nodes, left delimiter=(, right delimiter=)]
  (m) {
    t & t^2 & 0 & 0 \\
    0 & t^2 & 0 & 0 \\
    \node [fill=blue!20] {1}; & 0 & t^2 & t^2 \\
    0 & 1 & 0 & t \\
    0 & 0 & t & 0 \\
  };
  \node [above=5pt of m-1-1,rotate=90,anchor=west] (top1) {$r_1$};
  \node [above=5pt of m-1-2,rotate=90,anchor=west] (top2) {$r_2$};
  \node [above=5pt of m-1-3,rotate=90,anchor=west] (top2) {$r_3$};
  \node [above=5pt of m-1-4,rotate=90,anchor=west] (top2) {$r_4$};
  \node [left=12pt of m-1-1,anchor=east] (left1) {$x$};
  \node [left=12pt of m-2-1,anchor=east] (left2) {$y$};
  \node [left=12pt of m-3-1,anchor=east] (left3) {$z+ty$};
  \node [left=12pt of m-4-1,anchor=east] (left4) {$u$};
  \node [left=12pt of m-5-1,anchor=east] (left5) {$v$};
\end{scope}
\begin{scope}[shift={(7.9,0)}]
  \matrix [matrix of math nodes, left delimiter=(, right delimiter=)]
  (m) {
    0 & t^2 & -t^3 & -t^3 \\
    0 & t^2 & 0 & 0 \\
    \node [fill=blue!20] {1}; & 0 & t^2 & t^2 \\
    0 & 1 & 0 & t \\
    0 & 0 & t & 0 \\
  };
  \node [above=5pt of m-1-1,rotate=90,anchor=west] (top1) {$r_1$};
  \node [above=5pt of m-1-2,rotate=90,anchor=west] (top2) {$r_2$};
  \node [above=5pt of m-1-3,rotate=90,anchor=west] (top2) {$r_3$};
  \node [above=5pt of m-1-4,rotate=90,anchor=west] (top2) {$r_4$};
  \node [left=12pt of m-1-1,anchor=east] (left1) {$x$};
  \node [left=12pt of m-2-1,anchor=east] (left2) {$y$};
  \node [left=12pt of m-3-1,anchor=east] (left3) {$z+ty+tx$};
  \node [left=12pt of m-4-1,anchor=east] (left4) {$u$};
  \node [left=12pt of m-5-1,anchor=east] (left5) {$v$};
\end{scope}
\begin{scope}[shift={(12.3,0)}]
  \matrix [matrix of math nodes, left delimiter=(, right delimiter=)]
  (m) {
    0 & t^2 & -t^3 & -t^3 \\
    0 & t^2 & 0 & 0 \\
    \node [fill=blue!20] {1}; & 0 & 0 & t^2 \\
    0 & 1 & 0 & t \\
    0 & 0 & t & 0 \\
  };
  \node [above=5pt of m-1-1,rotate=90,anchor=west] (top1) {$r_1$};
  \node [above=5pt of m-1-2,rotate=90,anchor=west] (top2) {$r_2$};
  \node [above=5pt of m-1-3,rotate=90,anchor=west] (top2) {$r_3-t^2r_1$};
  \node [above=5pt of m-1-4,rotate=90,anchor=west] (top2) {$r_4$};
  \node [left=12pt of m-1-1,anchor=east] (left1) {$x$};
  \node [left=12pt of m-2-1,anchor=east] (left2) {$y$};
  \node [left=12pt of m-3-1,anchor=east] (left3) {$z'$};
  \node [left=12pt of m-4-1,anchor=east] (left4) {$u$};
  \node [left=12pt of m-5-1,anchor=east] (left5) {$v$};
\end{scope}
\begin{scope}[shift={(-0.5,-5)}]
  \matrix [matrix of math nodes, left delimiter=(, right delimiter=)]
  (m) {
    0 & t^2 & -t^3 & -t^3 \\
    0 & t^2 & 0 & 0 \\
    1 & 0 & 0 & 0 \\
    0 & \node [fill=blue!20] {1}; & 0 & t \\
    0 & 0 & t & 0 \\
  };
  \node [above=5pt of m-1-1,rotate=90,anchor=west] (top1) {$r_1$};
  \node [above=5pt of m-1-2,rotate=90,anchor=west] (top2) {$r_2$};
  \node [above=5pt of m-1-3,rotate=90,anchor=west] (top2) {$r_3-t^2r_1$};
  \node [above=5pt of m-1-4,rotate=90,anchor=west] (top2) {$r_4-t^2r_1$};
  \node [left=12pt of m-1-1,anchor=east] (left1) {$x$};
  \node [left=12pt of m-2-1,anchor=east] (left2) {$y$};
  \node [left=12pt of m-3-1,anchor=east] (left3) {$z+ty+tx$};
  \node [left=12pt of m-4-1,anchor=east] (left4) {$u$};
  \node [left=12pt of m-5-1,anchor=east] (left5) {$v$};
\end{scope}
\begin{scope}[shift={(5,-5)}]
  \matrix [matrix of math nodes, left delimiter=(, right delimiter=)]
  (m) {
    0 & t^2 & -t^3 & -t^3 \\
    0 & 0 & 0 & -t^3 \\
    1 & 0 & 0 & 0 \\
    0 & \node [fill=blue!20] {1}; & 0 & t \\
    0 & 0 & t & 0 \\
  };
  \node [above=5pt of m-1-1,rotate=90,anchor=west] (top1) {$r_1$};
  \node [above=5pt of m-1-2,rotate=90,anchor=west] (top2) {$r_2$};
  \node [above=5pt of m-1-3,rotate=90,anchor=west] (top2) {$r_3-t^2r_1$};
  \node [above=5pt of m-1-4,rotate=90,anchor=west] (top2) {$r_4-t^2r_1$};
  \node [left=12pt of m-1-1,anchor=east] (left1) {$x$};
  \node [left=12pt of m-2-1,anchor=east] (left2) {$y$};
  \node [left=12pt of m-3-1,anchor=east] (left3) {$z'$};
  \node [left=12pt of m-4-1,anchor=east] (left4) {$u+t^2y$};
  \node [left=12pt of m-5-1,anchor=east] (left5) {$v$};
\end{scope}
\begin{scope}[shift={(11,-5)}]
  \matrix [matrix of math nodes, left delimiter=(, right delimiter=)]
  (m) {
    0 & 0 & -t^3 & -2t^3 \\
    0 & 0 & 0 & -t^3 \\
    1 & 0 & 0 & 0 \\
    0 & \node [fill=blue!20] {1}; & 0 & t \\
    0 & 0 & t & 0 \\
  };
  \node [above=5pt of m-1-1,rotate=90,anchor=west] (top1) {$r_1$};
  \node [above=5pt of m-1-2,rotate=90,anchor=west] (top2) {$r_2$};
  \node [above=5pt of m-1-3,rotate=90,anchor=west] (top2) {$r_3-t^2r_1$};
  \node [above=5pt of m-1-4,rotate=90,anchor=west] (top2) {$r_4-t^2r_1$};
  \node [left=12pt of m-1-1,anchor=east] (left1) {$x$};
  \node [left=12pt of m-2-1,anchor=east] (left2) {$y$};
  \node [left=12pt of m-3-1,anchor=east] (left3) {$z'$};
  \node [left=12pt of m-4-1,anchor=east] (left4) {$u+t^2y+t^2x$};
  \node [left=12pt of m-5-1,anchor=east] (left5) {$v$};
\end{scope}
\begin{scope}[shift={(-0.5,-10.7)}]
  \matrix [matrix of math nodes, left delimiter=(, right delimiter=)]
  (m) {
    0 & 0 & -t^3 & -2t^3 \\
    0 & 0 & 0 & -t^3 \\
    1 & 0 & 0 & 0 \\
    0 & 1 & 0 & 0 \\
    0 & 0 & \node [fill=blue!20] {t}; & 0 \\
  };
  \node [above=5pt of m-1-1,rotate=90,anchor=west] (top1) {$r_1$};
  \node [above=5pt of m-1-2,rotate=90,anchor=west] (top2) {$r_2$};
  \node [above=5pt of m-1-3,rotate=90,anchor=west] (top2) {$r_3-t^2r_1$};
  \node [above=5pt of m-1-4,rotate=90,anchor=west] (top2) {$r_4-tr_2-t^2r_1$};
  \node [left=12pt of m-1-1,anchor=east] (left1) {$x$};
  \node [left=12pt of m-2-1,anchor=east] (left2) {$y$};
  \node [left=12pt of m-3-1,anchor=east] (left3) {$z'$};
  \node [left=12pt of m-4-1,anchor=east] (left4) {$u'$};
  \node [left=12pt of m-5-1,anchor=east] (left5) {$v$};
\end{scope}
\begin{scope}[shift={(5,-10.7)}]
  \matrix [matrix of math nodes, left delimiter=(, right delimiter=)]
  (m) {
    0 & 0 & 0 & -2t^3 \\
    0 & 0 & 0 & \node [fill=blue!20] {-t^3}; \\
    1 & 0 & 0 & 0 \\
    0 & 1 & 0 & 0 \\
    0 & 0 & t & 0 \\
  };
  \node [above=5pt of m-1-1,rotate=90,anchor=west] (top1) {$r_1$};
  \node [above=5pt of m-1-2,rotate=90,anchor=west] (top2) {$r_2$};
  \node [above=5pt of m-1-3,rotate=90,anchor=west] (top2) {$r_3-t^2r_1$};
  \node [above=5pt of m-1-4,rotate=90,anchor=west] (top2) {$r_4-tr_2-t^2r_1$};
  \node [left=12pt of m-1-1,anchor=east] (left1) {$x$};
  \node [left=12pt of m-2-1,anchor=east] (left2) {$y$};
  \node [left=12pt of m-3-1,anchor=east] (left3) {$z'$};
  \node [left=12pt of m-4-1,anchor=east] (left4) {$u'$};
  \node [left=12pt of m-5-1,anchor=east] (left5) {$v-t^3x$};
\end{scope}
\begin{scope}[shift={(11,-10.7)}]
  \matrix [matrix of math nodes, left delimiter=(, right delimiter=)]
  (m) {
    0 & 0 & 0 & 0 \\
    0 & 0 & 0 & \node [fill=red!20] {-t^3}; \\
    \node [fill=red!20]{ 1 } ;& 0 & 0 & 0 \\
    0 & \node [fill=red!20] {1}; & 0 & 0 \\
    0 & 0 & \node [fill=red!20] {t}; & 0 \\
  };
  \node [above=5pt of m-1-1,rotate=90,anchor=west] (top1) {$r_1$};
  \node [above=5pt of m-1-2,rotate=90,anchor=west] (top2) {$r_2$};
  \node [above=5pt of m-1-3,rotate=90,anchor=west] (top2) {$r_3-t^2r_1$};
  \node [above=5pt of m-1-4,rotate=90,anchor=west] (top2) {$r_4-tr_2-t^2r_1$};
  \node [left=12pt of m-1-1,anchor=east] (left1) {$x$};
  \node [left=12pt of m-2-1,anchor=east] (left2) {$y+2x$};
  \node [left=12pt of m-3-1,anchor=east] (left3) {$z'$};
  \node [left=12pt of m-4-1,anchor=east] (left4) {$u'$};
  \node [left=12pt of m-5-1,anchor=east] (left5) {$v'$};
\end{scope}

\draw[thick, ->] (0,0.3) -- (.7,0.3);
\draw[thick, ->] (4.2,0.3) -- (5.4,0.3);
\draw[thick, ->] (9.6,0) -- (10.1,0);

\draw[thick, ->] (1.5,-5) -- (2.5,-5);
\draw[thick, ->] (6.9,-5) -- (8,-5);

\draw[thick, ->] (1.5,-10.7) -- (2.5,-10.7);
\draw[thick, ->] (6.9,-10.7) -- (8,-10.7);

\draw[->,gray,rounded corners] (14,0) -- (14.3,0) -- (14.3,-1.8) -- (-3.5,-1.8) -- (-3.5,-4.6) -- (-3,-4.6);

\draw[->,gray,rounded corners] (12.7,-5) -- (13.2,-5) -- (13.2,-6.65) -- (-3.5,-6.65) -- (-3.5,-10.7) -- (-3,-10.7);

  \end{tikzpicture}

  \caption{Second example: graded Smith normal form for the
    persistence module given by $\langle x,y,z,u,v\rangle/\langle
    z+tx+ty, u+t^2x+t^2y, tv+t^2z+t^3y, tu+t^2z+t^3y\rangle$. The
    computation extracts the new basis elements $y'=y+2x$,
    $z'=z+ty+tx$, $u'=u+t^2y+t^2x$, $v'=v-t^3x$. In the final
    matrix, the elements corresponding to finite bars are shown:
    two bars of length 0, 1 of length 1 and one of length 3. The
    zero row corresponds to an infinite bar.  }
  \label{fig:module-gsnf}

\end{figure}

It is well worth noticing that the algorithms in existence (such as~\cite{cz2005,edelsbrunner2000topological,de_silva_dualities_2011,morozov_persistence_2005}) perform essentially the same tasks as we did using the Smith normal form to find a barcode for a homology computation; but taking significant shortcuts motivated by the particular cases the algorithms are built to deal with. In particular, most if not all algorithms only recover the pivot for the Smith normal form, and only encode the $\kk[t]$-module structure implicitly, in chosen basis orderings.

While this algorithm is certainly not the fastest option for
computing persistent homology, we believe it is the theoretically
easiest to understand. We include a more computer science
oriented presentation of the algorithm in Appendix~\ref{sec:algor-comp-pres}.

\subsection{Constructions}\label{sec:constructions}

\subsubsection{Presentation of Modules}\label{sec:presentation_modules}

Suppose $M$ is a finitely presented module over a ring $R$. Then, $M$
is given by some quotient $R^d/K$ for $K$ a finitely generated
submodule of $R^d$. If $R$ has global dimension at most 1, $K$ is a
free module, and thus has a finite basis expressed as elements of
$R^d$.

Hence, to track $M$ as a module, and to compute with elements of $M$,
it is enough to keep track of $d$ and of a basis of $K$. If the basis
of $K$ is maintained as a Gröbner basis, then we can compute normal
forms for any element of $M$ that have the property that if normal
forms of $m, n$ are equal, then the elements are equal in the quotient
module.

Quite often, we shall meet constructions where the module $M$ is more
naturally expressed as a subquotient of some semantically relevant
free module, possibly of higher than necessary rank. In this case, we
shall track two sets of data to enable computation in $M$: a basis of
the relations submodule $K$, and an extension of that basis to a basis
of the generators submodule $G$, resulting in two bases that track the
behaviour of $M=G/K$. Again, as long as $K$ has a basis maintained as
a Gröbner basis, it is enough to compute this normal form to compare
elements.

In particular, this is the case where the module $M$ is the homology of some chain complex. There, the free module is the module of chains, and the submodules $G$ and $K$ correspond to the modules of \emph{cycles} and \emph{boundaries} respectively.

We observe that this reflects practice among algorithms for
persistent homology: to compute the persistent homology with a
basis of a filtered simplicial complex, we maintain a cycle basis
and a boundary basis, reducing modulo the boundary at each step
to find out if the boundary $\partial\sigma$ of a newly
introduced simplex $\sigma$ is already a boundary or not.

\subsubsection{Useful Forms}\label{sec:useful_forms}

We have already seen the introduction of a \dfn{Smith normal form} to the graded $\kk[t]$-module context in Section \ref{sec:graded-smith-normal}. The Smith normal form transforms a matrix representing a map between free modules by basis changes in the source and target modules until the matrix is diagonal; all while respecting the grading of the modules.

Another fundamentally useful form that we will be using a lot is the \dfn{Row or Column echelon form} of a matrix. For a matrix $F$ representing a map $f:M\to N$ between free graded $\kk[t]$-modules, a row echelon form changes basis for $N$ to eliminate all redundancies in the information contained in the matrix, while a column echelon form changes basis for $M$ to achieve the same goal. Certainly, a Smith normal form is an echelon form, but the converse does not necessarily hold true.

The properties most interesting to us for an echelon form -- and we shall state these for the row echelon form; column echelon form holds mutatis mutandis --- are:
\begin{itemize}
\item All completely zero rows are at the bottom of the matrix
\item The leading coefficient of a non-zero row (the \dfn{pivot} of that row) is strictly to the right of the leading coefficient of the row above it
\item All entries in a column below a leading coefficient of some row are all zero.
\end{itemize}

The power of the echelon forms come in what they imply when you use a matrix to pick out a basis of a submodule: if the basis elements of a submodule of a free module are in the rows of a matrix, and this matrix is placed on row echelon form, then several problems concerning module membership and coefficient choice become easy to work on.

To determine if an element is in the submodule, add the element to the bottom row of the matrix, and put the matrix on row echelon form again. This, by force, will eliminate any entries in pivot columns of the new row; and modify elements to the right of the pivots that can reduce. Eliminating all column entries that occur as pivots of the previous basis puts the new element on a \dfn{normal form} with respect to the basis -- and this normal form is sufficient to determine equality modulo the submodule represented by the basis.

\subsection{Representation of Homomorphisms}\label{sec:rep_homomorphisms}

\textcite{cohen2009persistent} describe a method for
computing the kernel, image, and cokernel persistence for the
case of what they call compatible filtrations. For two
filtrations, $f$ over a space $X$ and $g$ over a space $Y$, the
two filtrations are compatible if $X \subseteq Y $ and the
restriction of $g$ to $X$ is equal to $f$. 

The algebraic equivalent (and generalization) of this is the notion of graded
maps.

\subsubsection{Dictionary of Operations}\label{sec:dictionary_operations}

\subsection{Two Modules}\label{sec:two_modules}
When given pairs of persistence modules, we must be able to represent
the maps between them, assuming of course, they follow the
assumptions in Section~\ref{sec:rep_homomorphisms}. This material
follows the treatment by \textcite[Chapter
  15]{eisenbud1995commutative}, specialized to the case of
persistence modules. We begin with presentations for modules $P$
and $Q$:

\begin{center}
  \begin{tikzcd}
    0 \arrow{r} & 
    G_P \arrow{r}{i_P} \arrow[dotted]{d}{\varphi|_{G_P}} & 
    F_P \arrow{r}{p_P} \arrow{d}{\varphi} &
    P \arrow{r} \arrow{d}{f} &
    0 \\
    0 \arrow{r} & 
    G_Q \arrow{r}{i_Q} & 
    F_Q \arrow{r}{p_Q} &
    Q \arrow{r} &
    0 \\  
  \end{tikzcd}
\end{center}

Such a presentation exhibits $P$ and $Q$ as quotients of free modules,
i.e. $P=F_P/i_PG_P$ and $Q=F_Q/i_QG_Q$. We call $F_P$ and $F_Q$
modules of generators of $P$ and $Q$ respectively, and $G_P$ and $G_Q$
modules of relations.

Over nice enough rings\footnote{Global dimension less than or
  equal to 1, which includes all PIDs and therefore all Euclidean
  domains}, both $G_P$ and $G_Q$ are free for all modules $P$. These
rings include all Euclidean domains.

Now, an arbitrary map $P\xrightarrow{f}Q$ can be represented by
a map from the generators of $P$ to the generators of $Q$, or in other
words a map $F_P\xrightarrow{\varphi}F_Q$. For such a map between the
generator modules to represent a map between the quotient modules it
needs to obey one condition: $\varphi \cdot i_P G_P \subseteq i_Q \cdot G_Q$. In other
words, any vector in $F_P$ that represents a relation in $P$, and
therefore represents the 0-element in $P$, has to be mapped to a
relation in $Q$ so that its image still represent the
0-element. Subject to this condition, any map between the free modules
$F_P$ and $F_Q$ will represent a map between $P$ and $Q$.

Among the core messages of this paper is the recognition that by
maintaining $F_P$ and $G_P$ we can compute just about anything to do
with persistence modules. In particular, if $F_P$ is kept in a nice
form -- so that all the elements of $G_P$ form a tagged sub-basis of
the basis of $F_P$, and all basis elements of $F_P$ not already in the
$G_P$ basis are kept \emph{reduced} with respect to $G_P$, then
questions about normal forms for elements of $P$, equality in $P$, and
many other relevant questions become easy to handle algorithmically.

Throughout the expositions below, it is universally helpful to keep
any module of relations on a row reduced echelon form and any module
of generators reduced and augmented by the relations. By doing this,
we can consider the effects on single basis elements of the various
constructions, and thereby the effects on a barcode presentation of a
module.

This kind of nested module presentation is to some extent inherent in
the persistence barcode: for the barcode to be accessible, the bases
for the persistence module expressed are kept in a format where any
finite barcode gives rise to a pair of basis elements -- one of $G_P$
that maps onto a basis element in $F_P$, with some grading shift
corresponding to the length of the barcode. By keeping the bases for
$F_P$ and $G_P$ synchronized so that all basis elements of $G_P$ form
halves of such pairs, the barcode becomes immediately connected to the
basis presentation.

To accomplish any useful computation with these representations
of persistence modules, we need to be able to have some basic
tools or manipulate them. In all cases, this will be done in
terms of some submodule of a joint quotient module. That is, by
considering both modules simultaneously and extracting the
relevant information we will end up in a situation where we input
presentations and as output get presentations, thus allowing
these tools to be composed together algorithmically. In addition,
we relate the matrix representations more familiar to
persistence literature to each case. 

\subsubsection{Chains or coefficients}
\label{sec:chains-or-coeff}

It is worth pointing out that depending on the application domain, two
different representations for the presentation may both be worth
while. One way to represent the generators and relations, especially
in an algebraic topological setting, is by giving explicit chains for
both in a chain module containing both. This way, the modules $G_P$
and $F_P$ are represented, essentially, by matrices representing the
maps $G_P\to C_*$ and $F_P\to C_*$ for the chain module $C_*$.

In this setting, algebraic manipulation on generators and relations
can be performed separately from each other. In particular, the
relations module can be set on a row echelon form independently of the
generators module. This is also the format emitted by existing
algorithms and implementations. One major drawback of this
representation is that the map $G_P\to C_*$ has to be pulled back to a
map $G_P\to F_P$ before a barcode can be computed.

Alternatively, the two maps $G_P\to F_P$ and $F_P\to C_*$ can be
maintained separately. In this setting, the components of a barcode
for $P$ are easy to extract; and it seems likely that the size of the
matrix $G_P\to F_P$ will be significantly smaller than a matrix
$G_P\to C_*$. However, each basis change in $F_P$ comes with the
requirement to modify two different matrices -- for each operation on
the matrix $F_P\to C_*$, the inverse operation has to be performed,
transposed, on the matrix for $G_P\to F_P$.

\subsubsection{Barcodes as module presentations}
\label{sec:barcodes-as-module}

For the particular case of a persistence barcode, the setting at hand
is very specific. The modules are $\mathbb N$-graded modules over
$\kk[t]$, and in the presentation of a persistence module $P$ as 
\[
0 \to G_P \xrightarrow{i_P} F_P \to P \to 0
\]
we maintain $G_P$ and $F_P$ as free submodules of a global \emph{chain
  module}, with a global specific basis given by the simplices in the
underlying simplicial complex. In this setting, we further maintain a
few invariants, in particular we ensure that $G_P$ and $F_P$ are the
results of the appropriate basis changes $S$ and $T$ from the Smith normal
form in order to guarantee that $i_P$ is the corresponding diagonal
matrix representation of the Smith normal form.

Hence, from an arbitrary presentation of a finitely presented
persistence module 
\[
0 \to G_P \xrightarrow{i_P} F_P \to P \to 0
\]
we get the barcode presentation by writing $\iota_P = Si_PT$ and then
replacing this presentation above with 
\[
0 \to T^{-1}G_P \xrightarrow{\iota_P} F_PS^{-1} \to P \to 0
\]

One concrete benefit of this approach is that in the new bases for
$T^{-1}G_P$ and $F_PS^{-1}$, we can match up basis elements between
the relations and the generators into birth/death pairs of barcode
intervals, and the corresponding diagonal entries in $\iota_P$ are
exactly on the form $kt^{\alpha}$ yielding a bar of length $\alpha$ in
the barcode.

\subsection{Direct Sum}\label{sec:direct-sum}
The first construction is the direct sum of $P$ and $Q$. This is
the simplest and most basic construction, but will be used in all
the other constructions. We construct it by taking the direct sum
across the relations and free part.
\[
0\to G_P\oplus G_Q 
\xrightarrow{\begin{pmatrix}i_P & 0 \\ 0 & i_Q\end{pmatrix}} 
F_P\oplus F_Q
\xrightarrow{\begin{pmatrix}p_P & 0 \\ 0 & p_Q\end{pmatrix}} 
P\oplus Q \to 0
\]
This is still a short exact sequence since with a direct sum the
maps can be composed as a direct sum of the component maps. 

In the case where the persistence module is represented as a matrix encoding a cycle basis and another matrix encoding a boundary operator, the new
representation is simply a block matrix with the component
matrices along the diagonal. Consider the matrix equations
\begin{eqnarray*}
R_P = D_P C_P\\
R_Q = D_Q C_Q
\end{eqnarray*}
The direct sum is given by
\begin{equation}
\begin{bmatrix}
R_P & 0\\0 & R_Q
\end{bmatrix}
=
\begin{bmatrix}
D_P & 0\\0 & D_Q
\end{bmatrix}
\begin{bmatrix}
C_P & 0\\0 & C_Q
\end{bmatrix}
\end{equation}

Assuming that the two inclusion maps $i_P$ and $i_Q$ are already on Smith normal form, the Smith normal form of the presentation map $i_P\oplus i_Q$ is given by the block matrix
\[
\begin{pmatrix}
  i_P & 0 \\ 0 & i_Q
\end{pmatrix}
\]

\subsection{Image}\label{sec:image}
Given a map, we will want to construct the image of the map
between the two modules. This is a relatively straightforward
construction. We first compute the image of the map between the
free parts: $\img \varphi : \img(F_P \rightarrow F_Q)$. The finite
presentation is then
\[
0\to G_Q
\xrightarrow{i_Q}
G_Q + \img\varphi
\to
\img(P\xrightarrow{f}Q) \to 0
\]

Note that the resulting image generators should be reduced with
respect to the relations to ensure that they are representative
generating elements.

It is sufficient to only consider $G_Q$ since we require $\varphi\cdot i_P 
G_P\subseteq i_Q\cdot G_Q$. If $\varphi$ is an inclusion, this implies
that $G_P\subseteq G_Q$. This is a general statement which
implies that if we have an inclusion, generators will map to
generators and boundaries map to boundaries.

In the context of this in topological terms as presented
in~\cite{cohen2009persistent}, we choose a basis for $F_P$ and $G_P$. By inclusion,
$\img \varphi = F_P$. Since $G_P \subseteq G_Q$, $G_P$ is
presented in a compatible basis for $F_p$ and so we simply extend
this basis to $G_Q$. Algorithmically this means finding computing
the persistent homology of $P$ then adding the simplices which
are in $Q$ and not in $P$ if they add any relations (positive
simplices are ignored as these would be in $F_Q$).

If we are starting with both module presentation on a barcode form, with the inclusion maps $i_P$ and $i_Q$ both on Smith normal form, then we can approach a Smith normal form for the presentation of the image by first computing the sequence of all $\varphi(f_j)$ for $f_j$ the basis of $F_P$ used.

Now consider the block matrix
\[
\begin{pmatrix}
  \begin{array}{c}
  \dots \varphi(f_1) \dots \\
  \vdots \\
  \dots \varphi(f_r) \dots \\
\end{array} & I_r & 0 \\
\begin{array}{c}
\dots i_P(g_1) \dots \\
  \vdots \\
  \dots i_P(g_s) \dots \\
\end{array}
& 0 & I_s \\
\end{pmatrix}
\]
with all of the images expressed in some basis of $F_Q$, and the right hand blocks consisting of an identity matrix in order to track the changes that will happen. Putting this matrix of a row reduced echelon form we will find a matrix on the shape
\[
\begin{pmatrix}
  \ast & \ast & 0 \\
  \ast & A & B
\end{pmatrix}
\]
where the blocks $(A B)$ carries the coefficients expressing the elements $i_P(g_j)$ in terms of the collection of elements $\varphi(f_j)$ and $i_P(g_j)$. In other words, this part is a presentation matrix for the map $G_Q\to G_Q+\img\varphi$ that we need. Putting this on Smith normal form yields barcodes for the image module.

\subsection{Cokernel}\label{sec:cokernel}

To compute a cokernel, we merely include the images of the
generators in $F_P$ among the relations for the cokernel as a
quotient module of $F_Q$.

\[
0\to G_Q \oplus F_P
\xrightarrow{i_P +  \varphi}
F_Q  \to 
\coker(P\xrightarrow{f}Q) \to 0
\]

To see why this is correct, first consider the cokernel of the
free modules. In this presentation if it is in the image, since
we have listed it among the relations, it maps to 0 in the
quotient module. One imaginable complication would be if $F_P$ has a basis element representing $0$ in $P$, but mapping to a basis element in $F_Q$ which is not killed by any element in $G_Q$. For a constellation like this, we would in the presentation above erroneously kill the generator in $F_Q$ by the presence in $F_P$; but the constellation is impossible, since $\varphi \cdot i_P \cdot G_P\subseteq i_Q \cdot G_Q$.

In terms of algorithms~\cite{cohen2009persistent}, we again compute a basis for
$F_P$ and $G_P$ just as in the image case. We now only consider
cycles and relations which do not appear in $P$.

\subsection{Kernel}\label{sec:kernel}

In order to construct the kernel for arbitrary finitely presented
modules, we first describe how to construct it for free
modules. Suppose we have a map between free modules $\psi: F\to
G$. Then the kernel of this map is the free submodule of $F$
consisting of elements that map to 0 in $G$. This can be computed
by the usual Gaussian elimination approach: we can write down a
matrix with rows in $G\oplus F$, each row a pair $(\psi(f), f)$
for $f$ a basis element of $F$. After putting this matrix on a
row-reduced echelon form, some of the rows will have entries only
in the $F$ part, all entries in the $G$ summand having reduced to
$0$. Projecting these rows back to $F$ gives us a basis for the
kernel.

The kernel of a map between finitely presented modules has to be
constructed in a two-step process. A generator of the kernel of
$P\xrightarrow{\varphi}Q$ is an element of $F_P$ such that its
image in $F_Q$ is also the image of something in $G_Q$. This is
to say that the free module of generators of
$\ker(P\xrightarrow{\varphi}Q)$ is given by the kernel of the map
between free modules $F_P\oplus G_Q\to F_Q$ given by
$(f,g)\mapsto(\varphi(f) - i_Q(g))$.

This takes care of the module $F_K$ of generators of
$\ker(P\xrightarrow{\varphi}Q)$. For a complete presentation we
also need the relations of this module. Since the kernel module
is a submodule of $P$, the relations are the restriction of the
relations in $P$ to the kernel $K$. These can be computed by
recognizing that they are given precisely by the elements in
$F_K$ such that their image of $F_K$ in $F_P$ coincides with
images of elements in $G_P$. Thus, we setup another map
$F_K\oplus G_P\to F_P$, defined by $(f,g)\mapsto f_P - i_P(g)$,
and compute the kernel of this map. This kernel is the module of
relations we need.

For this case, we cannot simply choose a basis and extend it in the
algorithmics (as the more complicated structure described above shows). There is no immediate way to take short cuts around the Smith normal form computation either.

Notice that if elements of the kernel are computed in some alternative
way, then the following lemma may be useful:

\begin{lemma}
  If $M$ is a free $\kk[t]$-module and $\phi: N\to M$ is
  $\kk[t]$-linear, then if $tz\in\ker\phi$, then $z\in\ker\phi$.
\end{lemma}
\begin{proof}
  Since $M$ is free, the only way that $tw=0$ in $M$ is if
  $w=0$. Since $tz\in\ker\phi$, we know that $\partial(tz)=0$. But by
  $\kk[t]$-linearity, $\partial(tz) = t\partial(z)$. Hence,
  $t\partial(z)=0$, so $\partial(z)=0$ follows.
\end{proof}

In particular, if $(C, \partial: C\to C)$ is a free $\kk[t]$-linear
chain complex, then if $tz$ is a cycle, then so is $z$.

\subsection{Free Pullback}\label{sec:free_pullback}
Before continuing into new constructions, we recount a useful
construction on free modules: the free pullback. Given two maps,
$f$ and $g$ in the following diagram:
\begin{center}
  \begin{tikzcd}
    & A\arrow{d}{f}\\
    B \arrow{r}{g}& C
  \end{tikzcd}
\end{center}
we construct the pullback $P$ such that the following diagram
commutes
\begin{center}
  \begin{tikzcd}
    P\arrow[dotted]{r} \arrow[dotted]{d}& A\arrow{d}{f}\\
    B \arrow{r}{g}& C
  \end{tikzcd}
\end{center}
To compute $P$, we must set up a kernel computation: namely find
a basis of $\ker(f\oplus -g$). Graphically, we construct the following matrix

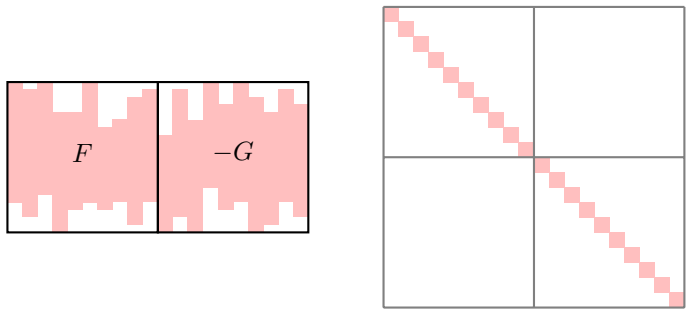
\begin{center}
  \begin{tikzpicture}[thick]
    
    \foreach \x/\y/\z in {0/0/1.6, 0.2/0.1/1.8, 0.4/0/1.5, 0.6/0.4/2, 0.8/0.4/1.7, 1/0/1.6} {
      \fill[pink] (-1+\x,1-\y) rectangle (-0.8+\x,1-\z);
    }

    \foreach \x/\y/\z in {1.2/0.6/1.7, 1.4/0.5/1.6, 1.6/0.2/1.9, 1.8/0.1/1.6, 2/0.7/2} {
      \fill[pink] (-1+\x,1-\y) rectangle (-0.8+\x,1-\z);
    }

    \foreach \x/\y/\z in { 2.2/0.1/1.8, 2.4/0.5/2, 2.6/0/1.4, 2.8/0.3/1.7, 3/0.3/1.6} {
      \fill[pink] (-1+\x,1-\y) rectangle (-0.8+\x,1-\z);
    }

    \foreach \x/\y/\z in {3.0/0/1.5, 3.2/0.2/2, 3.4/0.4/1.9, 3.6/0.1/1.6, 3.8/0.3/1.8} {
      \fill[pink] (-1+\x,1-\y) rectangle (-0.8+\x,1-\z);
    }

    \foreach \x in {0,0.2,...,3.8} {
      \fill[pink] (4+\x,2-\x) rectangle (4.2+\x,1.8-\x);
    }
    \draw[gray,step=2] (4,-2) grid (8,2);
    \node[coordinate,label=$F$] (a) at (0,-0.2) {$F$};
    \node[coordinate,label=$-G$] (a) at (2,-0.2) {$-G$};
    \draw[black] (-1,-1) rectangle (1,1);
    \draw[black] (1,1) rectangle (3,-1);

  \end{tikzpicture}

We perform a column reduction on the matrix $[F | -G]$ (on the
right) to compute a basis for the kernel. To keep track of the
basis we perform the same operations on the matrix on the right.
  
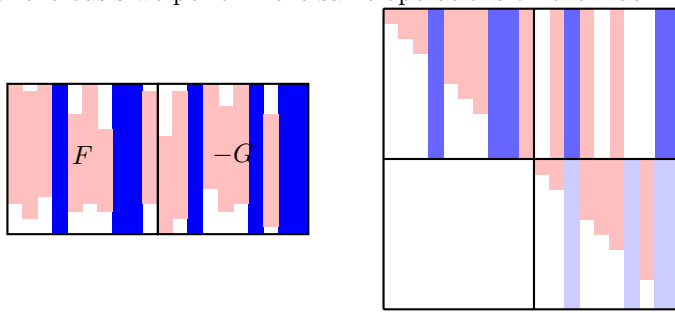
\begin{tikzpicture}[thick]
    
    \foreach \x/\y/\z in {0/0/1.6, 0.2/0.1/1.8, 0.4/0/1.5,  0.8/0.4/1.7, 1/0/1.6} {
      \fill[pink] (-1+\x,1-\y) rectangle (-0.8+\x,1-\z);
    }
    \foreach \x/\y/\z in {0.6/0.4/2,1.4/0.5/1.6, 1.6/0.2/1.9} {
      \fill[blue] (-1+\x,1) rectangle (-0.8+\x,-1);
    }

    \foreach \x/\y/\z in { 2.4/0.5/2,  3.2/0.2/2, 3.6/0.1/1.6, 3.8/0.3/1.8} {
      \fill[blue] (-1+\x,1) rectangle (-0.8+\x,-1);
    }

    \foreach \x/\y/\z in {1.2/0.6/1.7,  1.8/0.1/1.6, 2/0.7/2} {
      \fill[pink] (-1+\x,1-\y) rectangle (-0.8+\x,1-\z);
    }

    \foreach \x/\y/\z in { 2.2/0.1/1.8, 2.6/0/1.4, 2.8/0.3/1.7, 3/0.3/1.6} {
      \fill[pink] (-1+\x,1-\y) rectangle (-0.8+\x,1-\z);
    }

    \foreach \x/\y/\z in {3.0/0/1.5, 3.4/0.4/1.9} {
      \fill[pink] (-1+\x,1-\y) rectangle (-0.8+\x,1-\z);
    }

    \fill[pink] (4,2) rectangle (4.2,1.8) ; 
    \foreach \x  in {0.2,0.4,...,3.8} {
      \foreach \y in {0,0.2,...,\x}{
      \fill[pink] (4+\x,2-\y) rectangle (4.2+\x,1.8-\y) ; 
    }}
    \foreach \x  in {0.6,1.4, 1.6, 2.4, 3.6, 3.8} {
            \fill[blue!60] (4+\x,2) rectangle (4.2+\x,0) ; 
    }
    \foreach \x  in { 2.4, 3.2, 3.6, 3.8} {
            \fill[blue!20] (4+\x,0) rectangle (4.2+\x,-2) ; 
    }
    \foreach \x  in { 3.4,2,3.2, 2.8} {
      \fill[white] (4+\x,2) rectangle (4.2+\x,0) ; 
    }
    \draw[black,step=2] (4,-2) grid (8,2);
    \node[coordinate,label=$F$] (a) at (0,-0.2) {$F$};
    \node[coordinate,label=$-G$] (a) at (2,-0.2) {$-G$};
    \draw[black] (-1,-1) rectangle (1,1);
    \draw[black] (1,1) rectangle (3,-1);

  \end{tikzpicture}
\end{center}

Once computed, the matrix on the right has some zero columns,
shown above in blue. The corresponding columns in the matrix on
the right, represent a basis for the kernel.  The projection maps
are given by taking the appropriate part of the basis: the top
part (the darker blue) represents the projection onto the domain
of $F$ and the lighter blue in the lower part represents the
projection onto the domain of $G$. Note that the elements in the
top right quadrant, represent the kernel of $F$ and the elements
with only non-zero entries in the lower half represent a basis
for the kernel of $G$.

With this construction we can restate the kernel construction in
Section~\ref{sec:kernel} graphically. The generators of the kernel $F_K$ are
given by the pullback
\begin{center}
  \begin{tikzcd}
    F_K\arrow[dotted]{r}{\pi_P} \arrow[dotted]{d}{\pi_Q}& F_P\arrow{d}{\varphi}\\
    G_Q \arrow{r}{i_Q}& F_Q
  \end{tikzcd}
\end{center}
 The relations $G_K$ are given by a second pullback
\begin{center}
  \begin{tikzcd}
    G_K\arrow[dotted]{r}{i_K} \arrow[dotted]{d}{\pi_Q}& G_P\arrow{d}{i_P}\\
    F_K \arrow{r}{\pi_P}& F_P
  \end{tikzcd}
\end{center}
giving the presentation $G_K \xrightarrow{i_K} F_K \rightarrow
\ker(\varphi)$.

\subsection{Pullback}\label{sec:pullback}

Given module presentations
\begin{align*}
0\to G_P \xrightarrow{i_P} F_P \xrightarrow{p_P} P &\to 0 \\
0\to G_Q \xrightarrow{i_Q} F_Q \xrightarrow{p_Q} Q &\to 0 \\
0\to G_R \xrightarrow{i_R} F_R \xrightarrow{p_R} R &\to 0   
\end{align*}
and maps $P\xrightarrow{f} R$ and $Q\xrightarrow{g} R$ given
as maps $F_P\xrightarrow{\varphi} F_R$ and $F_Q\xrightarrow{\psi} F_R$
such that $\varphi(G_P)\subseteq G_R$ and $\psi(G_Q)\subseteq G_R$, we
can define the \dfn{pullback}. This is the submodule of $P\oplus Q$
such that for any element $(p,q)$, $\varphi p = \psi q$.

Pullbacks are useful when dealing with modules -- they can express
very many algebraic concepts neatly and concisely, such as inverse images
of maps, fibre products, kernels, and so on.

To compute the pullback, we set this up analogously as with the
free pullback, as a kernel computation. The pullback is a
submodule of $P\oplus Q$, and is carved out exactly as the kernel
of the map $P\oplus Q\to R$ given by $(p,q)\mapsto \varphi p-\psi
q$.  Since this is the kernel of a module map, we can compute
this using the construction of Section~\ref{sec:free_pullback}.

\begin{center}
  \begin{tikzcd}
    0 \arrow{r} & 
    G_P\oplus G_Q \arrow{r}{} \arrow[dotted]{d}{\varphi|_{G_P} -\psi|_{G_Q}} & 
    F_P\oplus F_Q \arrow{r}{} \arrow{d}{\varphi -\psi} &
    P \oplus Q  \arrow{r} \arrow{d}{f-g} &
    0 \\
    0 \arrow{r} & 
    G_R \arrow{r}{} & 
    F_R \arrow{r}{} &
    R \arrow{r} &
    0 \\  
  \end{tikzcd}
\end{center}

Expanding the construction, we compute the following two free pullbacks
\begin{center}
  \begin{tikzcd}
    F_{PB}\arrow[dotted]{r}{\pi_P} \arrow[dotted]{d}{\pi_Q}& F_P\arrow{d}{\varphi}\\
    F_Q \arrow{r}{i_Q}& F_R
  \end{tikzcd}
  \hspace{2cm}
  \begin{tikzcd}
    G_{PB}\arrow[dotted]{r}{\pi_{\oplus}} \arrow[dotted]{d}{i_{PB}}& G_P\oplus G_Q \arrow{d}{i_\oplus}\\
    F_{PB} \arrow{r}{\pi_P}& F_{P}\oplus F_Q
  \end{tikzcd}
\end{center}

Here we begin to see that with presentation algorithms can be
broken down into a few key constructions.

\subsection{Pushout}\label{sec:pushout}

Pushouts are the dual construction to pullbacks. Some of their most
important uses is to glue together things that overlap slightly --
producing almost but not quite a direct sum. 

Given module presentations 
\begin{align*}
0\to G_P \xrightarrow{i_P} F_P \xrightarrow{p_P} P &\to 0 \\
0\to G_Q \xrightarrow{i_Q} F_Q \xrightarrow{p_Q} Q &\to 0 \\
0\to G_R \xrightarrow{i_R} F_R \xrightarrow{p_R} R &\to 0   
\end{align*}
and maps $R\xrightarrow{f} P$ and $R\xrightarrow{g} Q$ given
as maps $F_R\xrightarrow{\varphi} F_P$ and $F_R\xrightarrow{\psi} F_Q$
such that $\varphi(G_R)\subseteq G_P$ and $\psi(G_R)\subseteq G_Q$, we
can define the \dfn{pushout}. 

This is the cokernel of the map $R\to P\times Q$ given by $r\mapsto
(\varphi r, -\psi r)$. This way, any image of an element $R$ in either
$P$ or $Q$ can ``move across the $\times$'' in the pushout module
$P\times_R Q$. 

\subsection{Tensor products}
\label{sec:tensor-products}

The tensor product, and its various associated constructions, are
easiest to describe if we focus on using the presentation map
$G\to F$ instead of explicit chain representations for the
relations. We assume that readers have seen tensor products at
least in the context of vector spaces and we recommend
~\textcite{eisenbud1995commutative} as a reference.

We fix a basis $B_P$ for $F_P$ and $B_Q$ for $F_Q$. A basis for
$M\otimes N$ is given by $B_P\otimes B_Q$. 

Relations are all generated by elements on the form $r\otimes g$ or
$g\otimes r$ where $r$ is a relation and $g$ is a generator.

\begin{theorem}\label{thm:tensorpresentation}
  Assume that the underlying coefficient ring is a graded PID.

  Suppose $P$ is presented by $i: G_P\to F_P$ given on Smith normal
  form, with $B_P$ the corresponding basis of $F_P$, and $Q$ is
  presented by $j: G_Q\to F_Q$ also given on Smith normal form, with
  $B_Q$ its corresponding basis of $F_Q$.

  Then a Smith normal form presentation of $P\otimes Q$ is given by a
  basis $B_P\times B_Q$ for the generating module, and a basis element
  in the relations module for each $p\otimes q$ for $p\in B_P, q\in
  B_Q$ where at least one of $\alpha p$ and $\beta q$ is a
  relation. Let $\gamma$ be the generator of the ideal $\langle
  \alpha, \beta\rangle$. Then $\gamma p\otimes q$ is the sole relation
  influencing $p\otimes q$ in the Smith normal form presentation of
  $P\otimes Q$.
\end{theorem}
\begin{proof}
  First, if all relations for both $P$ and $Q$ are given as Smith
  normal forms, any relation on the form $\alpha p\otimes q$ for $p, q$
  images of the relations basis are going to be products of basis
  elements; and thus not expand bi-linearly. Hence, each relation
  induced by the relations construction above is already a coefficient
  times a basis element.

  To ascertain a graded Smith normal form, we also need to verify that
  each basis element occurs in only one relation. This is not
  immediately guaranteed from the construction above -- there is one
  small correction step needed. Certainly, if $\alpha p\in iG_P$ and
  $\beta q\in jG_Q$, then $\gamma p\otimes q$ has to be a
  relation. However, if both $\alpha\neq 0$ and $\beta q\neq 0$, then
  both $\alpha p\otimes q$ and $\beta p\otimes q$ will show up from
  the construction above; for these cases, replacing the two
  candidates by their common generator $\gamma p\otimes q$ is needed
  to construct a Smith normal form presentation.
\end{proof}

We may note that the tensor product of the presentation is longer than
the presentation suggested above: it would be 
\[
0\to G_P\otimes G_Q \to G_P\otimes F_Q\oplus F_P\otimes G_Q\to F_P\otimes
F_Q\to P\otimes Q\to 0
\]

This is a \emph{free resolution} of $P\otimes Q$; but not a minimal
one --- the method in the proof of Theorem \ref{thm:tensorpresentation}
tells us exactly how to reduce away the redundancy represented by the
syzygies in $G_P\otimes G_Q$ to get a minimal presentation --- of
length corresponding to the homological dimension of the module
category.

For an arbitrary size tensor product $P_1\otimes\dots\otimes P_n$
with $P_j$ having relations $G_j$, generators $F_j$, Smith normal
form presentation map $i_j$, and a basis of $F_j$ denoted by
$B_j$, a basis for the generators is given by $n$-tuples in
$B_1\times\dots\times B_n$, and relations are given from the
generator of the ideal generated by all the coefficients of
relations of factor basis elements of the relation.

We can summarize the algorithm above as follows:

\noindent\textbf{Tensor Product Algorithm}\\
\noindent\emph{Input}: Two module presentations $P$ and $Q$ with
reduced bases in Smith Normal Form
\begin{enumerate}
\item  The generators are the tensor product of the generators: $F_{P\otimes Q} = F_P \otimes F_Q$ --- generating elements are of the form  $(f_P(i),f_Q(j))$ for all generator elements in $P$ and $Q$
\item   Create a non-minimal list of relations by pairing relations and generators --- $G_{P\otimes Q} = F_P \otimes G_Q \oplus F_Q \otimes G_P$ 
\item Create a minimal representation of relations --- reduce
the relations modulo $G_{P}\otimes G_Q$ --- consider the pairs
  $(g_p,f_Q)$ and $(f_p,g_Q)$ such that $g_p\rightarrow f_p$ and
  $g_Q\rightarrow f_Q$ and only keep the relation which occurs first. 
\end{enumerate}
\noindent\emph{Output:} A presentation of the tensor product

\begin{algorithm}
\caption{Tensor Product Algorithm}
\begin{algorithmic}[1]
\STATE Input: Two module presentations $P$ and $Q$ with
reduced bases in SNF:\\ 
$G_P \rightarrow F_P  \rightarrow P$\\
$G_Q \rightarrow F_Q  \rightarrow Q$
\STATE Compute generators: $F_{P\otimes Q} = F_P \otimes F_Q =  (f_P(i),f_Q(j)) \forall f_P(i) \in F_P, f_Q(j) \in F_Q$
\STATE Compute non-minimal relations: $G_{P\otimes Q} = F_P \otimes G_Q \oplus F_Q \otimes G_P = (f_P(i),g_Q(j))\oplus (g_P(k),f_Q(\ell))\qquad \forall f_P(i) \in F_P, f_Q(\ell) \in F_Q, g_P(j) \in G_P, f_Q(k) \in F_Q$,  
\STATE For pairs $(g_p,f_Q)$ and $(f_p,g_Q)$ such that $g_p\rightarrow f_p$ and  $g_Q\rightarrow f_Q$, keep the relation which occurs first. 
\STATE{Output:} Relations, generators and a map:
$G_{P\otimes Q} \rightarrow F_{P\otimes Q}$
\end{algorithmic}
\end{algorithm}

We give an example of this process in Section~\ref{sec:relat-hom-modul}.

\subsubsection{Interpretations of tensor products}
\label{sec:interpr-tens-prod}

The tensor product described above is the tensor product
$M\otimes_{\kk[t]}N$, in which the $t$ parameter can be freely moved
``across'' the $\otimes$ sign. One can consider this to be a
tensor-product of persistence modules, in which time steps for the
conglomerate by stepping time arbitrarily on one side or the other.

There are more tensor products we can construct: by forgetting about
the persistence structure on one or the other side -- or by having a
persistence module $M$ and a $\kk$-vector space $N$ -- we can produce
tensor products $M\otimes_\kk N$. For these, the $t$ action does not
move fluidly between different factors of the tensor product; but we
can still recover a $\kk[t]$-module structure on the tensor product.

In fact, if we have persistence modules $M, N$, then
$M\otimes_{\kk}N$ has two different $\kk[t]$-module structures; these
are defined on the basis elements by
\[
t(m\otimes n) = tm\otimes n 
\qquad\qquad
t(m\otimes n) = m\otimes tn 
\]
and extended by linearity to the entire tensor product. We'll explain
the first of these structures -- the second one is completely
analogous. In this case, different persistence parameter values for
$N$ give completely different components of the product: stepping with
time only moves up on the $M$ side of the factor.

These constructions correspond to the different dualizing functors
introduced by~\textcite{de_silva_dualities_2011}. 

\subsubsection{Relations to $\hom$ modules}
\label{sec:relat-hom-modul}

Suppose $M$ is a persistence module. Then $M^*=\hom(M,
\kk[t])$ is a persistence module too, whose degree $d$ elements are
maps $M\to\kk[t]$ that change degree by $d$; so that if $|m|=n$ and
$|f|=d$ then $|f(m)|=n+d$.

This module has the structure of a graded $\kk[t]$-module by the rule
$(tf)(m) = f(tm)$, and we can easily verify that this $(tf)$ has an
increased degree as defined above.

Given a basis $B=\{m_1,\dots,m_k\}$ of $M$, there is a \emph{dual
  basis} $B^*=\{m^*_1,\dots,m^*_k\}$ of $M^*$ where $m^*_i(m_j)$ is the
Dirac delta function $\delta_{ij}$.

We consider the module $\hom(M, N)$ for persistence modules $M$,
$N$ consisting of maps $f: M\to N$ of arbitrary degree, just as
defined above: $f$ has degree $|d|$ if when $|m|=n$ then
$|f(m)|=d+n$.

\textbf{Example:} Suppose $M$ is given by two generators $x, y$ in degrees 1 and 2, and
the relations $t^3x, t^4y$. Suppose that $N$ is given by three
generators $u,v,w$ in degrees 1,1,2, and with relations $t^2u, tv,
t^4w$. The dual module $M^*$ has generators $x^*, y^*:M\to\kk[t]$
defined by $x^*(x) = 1, x^*(y)=0, y^*(x)=0, y^*(y)=1$. Thus, the
generators have degrees $-1$ and $-2$ respectively. We can establish
by inspection that $t^3x^*=0$ and $t^4y^*=0$, producing the relations
of $M^*$.

A generating set for $M^*\otimes N$ is given
by all pairs of elements: $x^*\otimes u, x^*\otimes v, x^*\otimes
w,y^*\otimes u, y^*\otimes v, y^*\otimes w$. Relations are all pairs
of generator with relation, thus given by the following list:
\begin{align*}
  t^3x^*\otimes u && 
  t^3x^*\otimes v && 
  t^3x^*\otimes w &&
  t^4y^*\otimes u && 
  t^4y^*\otimes v && 
  t^4y^*\otimes w \\
  x^*\otimes t^2u &&
  x^*\otimes tv &&
  x^*\otimes t^4w &&
  y^*\otimes t^2u &&
  y^*\otimes tv &&
  y^*\otimes t^4w \\
\end{align*}
Out of each column, the lowest degree of $t$ remains as the actual
relation, giving us the final list $t^2\cdot x^*\otimes u, t\cdot
x^*\otimes v, t^3\cdot x^*\otimes w, t^2\cdot y^*\otimes u, t\cdot
y^*\otimes v, t^4 y^*\otimes w$. 

One persistence homomorphism $f:M\to N$ is given by $f(x)=u+v$,
$f(y)=tu+w$. In the presentation above, this homomorphism is given by
the linear combination
\[
f = x^*\otimes u+x^*\otimes v+t\cdot y^*\otimes u+y^*\otimes w
\]
and by computing degrees of each term, summing degrees for the basis
elements in each pair, we can see that $f$ has degree 0 as a
homomorphism.

\subsection{Exterior powers}
\label{sec:external-powers}

The exterior power $\Lambda^2(M)$ is given by $M\otimes M/\sim$ where
$a\otimes b \sim -b\otimes a$. We get it by adding $f\otimes g +
g\otimes f$ for all generators $f, g$ of $M$ to the relations module
of the presentation of the tensor product $M\otimes M$.

We can choose a different
representation and handle the exterior powers. The appropriate
representation here is to assign as a generator of the exterior power
$\Lambda^m(M)$ a \emph{set} of $m$ distinct elements from the
component basis elements, paired with a permutation $\sigma\in S_m$
tracking the ``original'' order of the basis elements. This
permutation is important for sign handling when comparing different
elements -- $(M, \sigma) = (M, \tau)$ iff $\sigma\tau^{-1}$ is an even
permutation; if $\sigma\tau^{-1}$ is an odd permutation, then
$(M,\sigma) = -(M, \tau)$. 

\subsubsection{Discrimination with exterior powers}
\label{sec:discr-with-extern}

The exterior power construction has some discriminating power between
possible barcodes, even without computing a Smith normal form
directly; just computing the exterior power can in some cases
distinguish between different possibilities for barcodes.

As an example, consider a persistence module $M$ with two finite
barcodes. It has generators in degrees $1, 2$ and relations in degrees
$5, 10$. Just from degree reasons, we know that the possible barcodes
are either $(1,5); (2,10)$ or $(1,10); (2,5)$.

Call the generators $x, y$; the exterior power $\Lambda^2M$ has as its
generating set the set of all cardinality 2 subsets of $\{x,y\}$,
which is to say the only basis element is $x\wedge y$.

As for the relations, a relation $at^dx+bt^ey$ gives rise to
relations $at^d(x\wedge y)$ and $bt^e(x\wedge y)$; which together
have as their basis $ct^f(x\wedge y)$ where $ct^f$ is the
higher-degree choice out of $at^d$ and $bt^e$. The higher degree
choice is really just the gcd and is  due to the nature of basis changes in these modules; if
$x$ exists earlier than $y$, then $y$ can be replaced by
$y-\lambda t^kx$, but $x$ cannot be replaced by $x-\lambda
t^{-k}y$. Since the relation is homogeneous, the earlier existing
basis element is the one with the higher degree power of $t$
associated.

The presentation matrix for $M$ will be
\[
\begin{pmatrix}
  at^4 & ct^9 \\ bt^3 & dt^8
\end{pmatrix}
\]

This yields a presentation matrix for $\Lambda^2M$ that has entries in
degree order
\[
\begin{pmatrix}
  -at^4 \\ at^4 \\ bt^3 \\ -bt^3 \\ -ct^9 \\ ct^9 \\ dt^8 \\ -dt^8
\end{pmatrix}
\]
with Smith normal form dependent on which of $a, b, c, d$ are
non-zero. In particular, if $b\neq 0$, then $t^3x\wedge y$ is a
relation, and if $b=0$ then $t^4x\wedge y$ is a relation.

Consider now the two cases we gave here. In the first option, our
relations are $t^4x, t^8y$, and in the second they are $t^9x,
t^3y$. If the relation in the exterior square is $t^4 x\wedge y$ or
$t^3 x\wedge y$, this tells us which of the two options applies to our
case.

This agrees with what we would have expected from a computation of a
barcode for $M$ before we computed the exterior algebra. Indeed, if
$b\neq 0$, then the Smith normal form computation of the presentation
for $M$ would have started by a change of basis, replacing $y$ by
$y'=y+a/b tx$. As a result, the relation transforms into $bt^3y'$,
giving a barcode $(2,5)$ which settles the entire module structure.

\subsection{Symmetric powers}
\label{sec:symmetric-powers}

The symmetric power $S^2(M)$ is given by $M\otimes M/\sim$ where
$a\otimes b \sim b\otimes a$. We get it by adding $f\otimes g -
g\otimes f$ for all generators $f, g$ of $M$ to the relations module
of the presentation of the tensor product $M\otimes M$.

Alternatively, one may represent the symmetric power $S^m(M)$ with a basis
corresponding to \emph{multisets} of weight $m$ of basis elements. The same
construction holds for finding relations as for the tensor products --
with the generating element of the coefficients that kill each
factor basis element in the multiset killing the compound basis element.

As opposed to exterior powers and the example in
Section~\ref{sec:discr-with-extern}, symmetric powers force
commutativity instead of anti-commutativity. Where $\Lambda^k(M)$
eventually vanishes for finite dimensional $M$, $S^k(M)$ is finite
dimensional and non-trivial whenever $M$ is.

\section{Applications}
\label{sec:applications}
Although our primary goal was to describe a general framework for
manipulating persistence modules in a common framework, there
are several application areas that have motivated the
constructions in this paper and show great promise in
applicability for the algorithms and ideas present. We shall
discuss some of these here.

\subsection{Torsion complexes and persistent relative homology}
\label{sec:pers-rel-homol}

Consider a persistent chain module with torsion -- a $\kk[t]$-module with some non-trivial relations, and with an endomorphism $\partial$ of degree 0 that squares to 0. Such modules arise in various contexts -- as part of stepwise approximation in spectral sequence approaches to persistent homology, as well as when computing relative homology in a persistent setting. In the former, torsion comes from the fact that we are dealing with persistence modules and in the latter comes from the quotienting at a space level.  

Now, the persistent homology of this directed system $C_*^*$ is
simply homology with coefficients in $\kk[t]$ of the torsion
persistence module.  In other words, we can represent torsion at
the chain level and compute homology as we normally would using
the techniques described earlier in this paper; we have all the
tools necessary for the computation of persistent homology of a
torsion chain complex. For a complete example of computing
persistent relative homology see Appendix~\ref{sec:an-example}.

\subsection{Spectral sequences}
\label{sec:spectral-sequences}

This paper was originally motivated by our work on spectral sequences in persistent homology\cite{lipsky_spectral_2011}. For our interest in spectral sequences where each entry is a persistence module, all fundamental algebraic constructions need to be accessible -- to compute subsequent pages of the spectral sequence, one repeatedly computes homologies locally at each point $E^k_{p,q}$ of a boundary map that successively approximates an end-result closer and closer.

For the first step, going from $E^0$ to $E^1$, the modules can usually be assumed to be free modules, and the justifications by~\textcite{cz2005} are enough to assert computability; but even for spectral sequences going from $E^1$ to $E^2$ and beyond, a firm grasp of the computation of kernels and cokernels of persistence modules with possible torsion is needed.

With the algorithms stated here, we are able to produce code for computing spectral sequences with persistence modules. These spectral sequences could be bases for several interesting directions for further study: parallelization by using a spectral sequence to merge local information into global, as described in~\cite{lipsky_spectral_2011}, or computing homology of the total space of a combinatorial fibration using fibres and base space by adapting the Serre spectral sequence as described by~\textcite{goerss_simplicial_1999} are two approaches that come to mind.

\subsection{Unordered persistence computation}
\label{sec:unord-pers-comp}

Many of the most efficient algorithms for constructing filtered complexes in the first place, such as the algorithms proposed by \textcite{cz2005}, produce simplices in some order compatible with the inclusion ordering on the simplices, but not necessarily compatible with the chosen filtration ordering. As an alternative to re-sorting the simplices this algebraic approach offers us algorithms to compute with an essentially unordered simplex sequence, performing reductions as needed, based on the new simplices acquired.

If the critical values of the filtration function are not known before hand, everything in this section can still be performed as long as one is prepared to adjust degrees for all elements in play along the way; for exposition, we shall assume that each simplex has a global $\mathbb{N}$-graded degree associated to it.

The classical persistence algorithm as described by \textcite{edelsbrunner2000topological} and by \textcite{cz2005} works by computing a Gaussian elimination on the boundary matrix, respecting the grading on the rows and columns involved by the choice of ordering.

When introducing a new boundary basis element, this will affect the current state of the Gaussian elimination running under the hood by first reducing the new basis element with everything that was already there -- by filtration value -- when the new element showed up, and then using the thus reduced element to further reduce any later additions to the boundary basis. This effect of the grading is transparent enough that we can fit it in a variation of the original persistence algorithm.

Adding a new simplex $\sigma$, we compute $\partial\sigma$, and reduce $\partial\sigma$ by all the boundaries that were added at earlier filtration values than that of $\sigma$. If the result is 0, then clearly, $\partial\sigma$ was a boundary, and so $\sigma$ is the leading simplex of  a cycle.

If $\partial\sigma$ does not vanish when reduced modulo the current boundary basis, then the leading simplex of the reduction of $\partial\sigma$ is some simplex $\tau$ associated to the cycle basis. Thus, $\sigma$ is a pre-boundary, bounding the cycle that has $\tau$ as its leading term. We introduce the interval $(f(\tau),f(\sigma))$ as a new interval.

Invariant of the persistence algorithm: at any stage, the simplices are paired up in a partial matching such that if $(\tau,\sigma)$ is a pair, then $\partial\sigma$ reduces modulo all the boundaries existing at the time of the introduction of $\sigma$ to a chain with leading simplex $\tau$.

This invariant can change with the introduction of some new simplex $\upsilon$ in the following ways:
\begin{itemize}
\item $\partial\upsilon$ reduces modulo earlier boundaries to 0. Thus, $\upsilon$ starts a new cycle.
\item $\partial\upsilon$ reduces modulo earlier boundaries to some cycle with leading simplex $\theta$. 
  \begin{itemize}
  \item If $\theta$ is the leading simplex of some later boundary basis element, linked to a generating pre-boundary simplex $\tau$, then there is some pair $(\theta,\tau)$. If $f(\upsilon) < f(\tau)$, then this pair is modified to $(\theta, \upsilon)$ and the boundary linked to $\tau$ is further reduced using $d\upsilon$. 

    This may well lead to cascading changes, as the resulting new boundary chains in turn reduce later and later chains, and replace more and more paired simplices.
  \item Otherwise, $\theta$ is an unpaired cycle simplex (since the ordering \emph{is} compatible with inclusion of simplices), and therefore creates a new simplex pairing and a new interval in the barcode.
  \end{itemize}
\end{itemize}

\subsubsection{Example}
\label{sec:example}

Take the 2-simplex spanned by $1,2,3$ and introduce the simplices in the following order, with the filtration values given as well:
\[
\begin{array}{r|ccccccc}
  \text{Simplex} & 1 & 2 & 12 & 3 & 13 & 23 & 123 \\ \hline
  \text{Filtration value} & 1 & 4 & 6 & 2 & 3 & 5 & 7
\end{array}
\]

At the start of the algorithm, nothing is out of the ordinary -- as long as filtration values increase monotonically, this works exactly as the classical persistence algorithm would. Hence, after consuming 3 simplices, we have a cycles basis consisting of $1$, and a boundary basis consisting of $2-1$, and a pairing of $(2,12)$.

Then arrives the simplex $3$. We compute $\partial(3) = 0$, and add $3$ to the cycles basis.

Next comes the simplex $13$. The boundary is $3-1$, and we pair it with $3$, yielding the pairings $(2,12), (3,13)$. At this stage, we have a cycles basis of $1$ and a boundary basis of $2-1, 3-1$. It is worth noting at this point that everything involved knows of its filtration value, and therefore of its own grading.

Now, the simplex $23$ arrives. The 1-simplices are ordered by their filtration values as $f(13) < f(23) < f(12)$, which concretely means that $\partial(23) = 3-2$ only gets reduced by $3-1$, to the value $-2+1 = -\partial(12)$. Further reductions are impossible with the older parts of the boundary basis, and thus this simplex \emph{should have originally paired} with its leading element. We adjust the pairing $(2,12)$ to $(2,23)$. Furthermore, the previously not entirely reduced element of the boundary basis given by $\partial(12) = 2-1$ now is reducible with this new simplex information; at which point $12$ transforms from a pre-boundary to a cycle.

Finally, the introduction of $123$ kills the cycle $12$ as expected.

\section{Conclusions \& Future Work}\label{sec:conclusions--future}

The primary motivation is to provide a general framework for more
complex operations with persistence modules --- hopefully
reducing the need to design new algorithms from scratch for new
constructions. The trade-off is that these are not optimal in
terms of running time or space but provide the flexibility for
experimentation.

We have demonstrated how techniques from commutative and
computational algebra influence the study and the computation of
persistent homology. We have shown how the classical concept of
Smith normal form can be adapted to graded modules over a graded
PID, and how this adaptation leads to concrete algorithms for
computing with persistent barcodes, and maps between persistence
modules. The common representation we use is the presentation of
a finitely generated module. By making the representation
explicit, we are able to state algorithms which take as input
presentations and output presentations, allowing for easy
composition of different constructions.

It would be interesting to extend the understanding of the interplay
between algebra and topology -- understand how different categories of
modules, different representation theories affect the computation of
persistence. We have already seen three categories that are useful for
different approaches to persistence -- modules over $\kk[t]$, modules
over $\kk Q$ for a quiver of type $A_n$. Furthermore, the representation theory of
$\kk[s,t]$ and the classification by
\textcite{gabriel1972unzerlegbare} of representations into tame and
wild types were instrumental to proving the limitations of
multidimensional persistence. 

\printbibliography

\appendix
\section{Algorithm for Computing Presentations of Persistence Modules}
\label{sec:algor-comp-pres}

In this appendix we give a more explicit algorithmic description
of Section~\ref{sec:graded-smith-normal}. We begin by describing computing standard
persistent homology.

Given a filtered simplicial complex, we write down its graded
boundary operator. The added needed data structure which is a
data structure which returns the relative grading between two
simplices. If we have a discrete parameter filtration (or if we
discretize a continuous parameter), then the relative grading is
the time difference between when the simplices appear. 

Therefore we can write down the boundary operator matrix using
the standard formula
\begin{equation}
\partial(\sigma) = \sum_i (-1)^{i} [v_1,v_2,v_{i-1},v_{i+1},v_n]
\end{equation}
where the $v_i$'s are the vertices of the simplex $\sigma$.

For each non-zero entry, we multiply it by the relative grading
of the two corresponding simplices. Since we are in chain space,
each row/column pair corresponds to two simplices. We adopt the
following conventions: we sort the row space such that each
column of the matrix has decreasing powers of $t$. This is
equivalent to sorting the simplices in order of appearance in the
filtration. We sort the columns, such that each row has
increasing powers of $t$. This also corresponds to sorting the
columns in terms order of appearance. 

We now introduce the standard reduction algorithm. We use the
following notation
\begin{itemize}
\item $F$ - a graded $m\times n$ matrix such that columns have decreasing powers of $t$ and rows have increasing powers of $t$
\item $F_i$ - $i$-th column vector
\item $F_{i,j}$ - $(i,j)$-th entry in $F$
\item $pivots(\cdot)$ -  a lookup table where we store pivot elements, it returns the column index which contains the pivot element of that row; if empty returns $\emptyset$
\item $diff(\cdot,\cdot)$ - a function which return the relative grading difference between the two inputs
\item $C$ - matrix which stores chains
\item $Z$ - kernel basis of $F$ (cycle basis)
\item $B$ - image basis of $F$ (boundary basis)
\end{itemize}

When computing persistence we set $F$ to the graded boundary matrix $\partial$. 
\begin{algorithm}[h]
\caption{Reduction Algorithm}
\begin{algorithmic}[1]
\STATE Input: $F = \partial$ 
\STATE Initialize $pivots=\emptyset$,  $C=I$ ($n\times n$ identity matrix ),$Z=\emptyset$, and $B=\emptyset$
\FOR{i=1,\ldots,n}
\WHILE{$F_i\neq 0$}
\STATE Set $j\leftarrow$ lowest non-zero entry (smallest power of $t$)
\IF{$pivots(j)\neq \emptyset$}
\STATE $X =  diff(pivots(j),j)$
\STATE $F_i = F_i - t^X (F_{i,j}/F_{i,pivots(j)}) F_{pivots(j)}$
\STATE $C_i = C_i - t^X (F_{i,j}/F_{i,pivots(j)}) C_{pivots(j)}$
\ELSE
\STATE Set $pivots(j) \leftarrow i$
\STATE Add $F_i$ to $Z$
\STATE Break
\ENDIF
\IF{$F_i = 0$}
\STATE Add $C_i$ to $Z$
\ENDIF
\ENDWHILE
\ENDFOR
\end{algorithmic}
\end{algorithm}

Note that this gives the basis elements of $Z$ and $B$ in order
of appearance. Each element is represented as a graded chain. As
in the example in Section~\ref{sec:graded-smith-normal}, we now
must write down the map $B\rightarrow Z$.  To obtain the barcode,
we must reduce the map to SNF.

The reduction algorithm of the map is almost identical to the
algorithm above except that once a pivot is found, we reduce the
rows above it and we do not need to store the change of basis if
we only care about the barcode. We use the same notation as
above, except $F: B\rightarrow Z$ and since we will need row
operations as well, $F^i$ represents the $i$-th row of $F$.
\begin{algorithm}[h]
\caption{Smith Normal Form Algorithm}
\begin{algorithmic}
\STATE Input $F: B\rightarrow Z$
\STATE $pivots=\emptyset$
\FOR{$i=1,\ldots,n$}
\WHILE{$F_i\neq 0$}
\STATE Set $j\leftarrow$ lowest non-zero entry (smallest power of $t$)
\IF{$pivots(j)\neq \emptyset$}
\STATE $X =  diff(pivots(j),j)$
\STATE $F_i = F_i - t^X (F_{j,i}/F_{j,pivots(j)}) F_{pivots(j)}$
\ELSE
\STATE Set $pivots(j) \leftarrow i$
\FOR{$k=j-1,\ldots,1$}
\STATE $Y =  diff(j,k)$
\STATE  $F^k = F^k  - t^Y (F_{k,i}/F_{k,j}) F^j$
\ENDFOR
\ENDIF
\ENDWHILE
\ENDFOR
\STATE All the pivots represent bars beginning at the grading of the corresponding generator and lasting the grading of the entry
\STATE All zero columns represent infinite bars beginning at the grading of the corresponding generator 
\end{algorithmic}
\end{algorithm}

To put this into context of the algorithm described in
\textcite{}. As above it computes cycles in order of appearance
and stores them by only keeping track of the element with the
lowest power of $t$. When computing boundaries, it only considers
the elements corresponding to elements in the cycle
basis. Therefore, the pivot already corresponds to an entry in
the SNF, since it represents a map from a boundary element to a
cycle (represented by one element --- the last simplex which
``created'' the cycle). While this is more efficient, it makes it
more difficult to generalize since we do not always have the
freedom to choose such nice bases .

To generalize the above algorithm, we must have the ability to
find relative gradings between chains. In the above case, this
can be done by mapping back to simplices.  In general with more
complicated operations, there may be many mappings to get back to
chain space. However, due to homogeneity, it is sufficient if we
can compare two general elements of the row space and column
space. Therefore in the above example, as we compute a basis, we
would build an associated structure which given two boundaries or
two cycles, returns their relative grading. 
 
\section{Relative Persistent Homology Example}
\label{sec:an-example}

Consider the 2-simplex, with each face introduced at a separate time. The extended homology of this is persistent homology of the persistent chains of this complex, with each simplex added to the torsion relations in inverse order. Faces are introduced at times $0$ through $6$, and added to the torsion relations at times $7$ through $13$.

The chain complex has the finite presentation 
\[
C=\frac{\sigma_0,\sigma_1,\sigma_2,\sigma_{01},\sigma_{02},\sigma_{12},\sigma_{012}}
{t\sigma_{012},t^3\sigma_{12},t^5\sigma_{02},t^7\sigma_{01},t^9\sigma_{2},t^{11}\sigma_{1},t^{13}\sigma_{0}}
\]

The boundary operation defined on $C$ is the normal boundary operation on simplicial complexes, defined by:
\begin{align*}
  \partial\sigma_0 &= 0 & \partial\sigma_{01} &= t^3\sigma_0-t^2\sigma_1 \\
  \partial\sigma_1 &= 0 & \partial\sigma_{02} &= t^4\sigma_0-t^2\sigma_2 \\
  \partial\sigma_2 &= 0 & \partial\sigma_{12} &= t^4\sigma_1-t^3\sigma_2 \\
  && \partial\sigma_{012} &= t^3\sigma_{01}-t^2\sigma_{02}+t\sigma_{12}
\end{align*}

To compute the kernel module, we first compute its generators by doing Gaussian elimination on the matrix to the left, yielding the result to the right:
\[
\left(
  \begin{array}{c|c|c}
    0 & \sigma_0 & 0 \\
    0 & \sigma_1 & 0 \\
    0 & \sigma_2 & 0 \\
    t^3\sigma_0-t^2\sigma_1 & \sigma_{01} & 0 \\
    t^4\sigma_0-t^2\sigma_2 & \sigma_{02} & 0 \\
    t^4\sigma_1-t^3\sigma_2 & \sigma_{12} & 0 \\
    t^3\sigma_{01}-t^2\sigma_{02}+t\sigma_{12} & \sigma_{012} & 0 \\
    t\sigma_{012} & 0 & \rho_0 \\
    t^3\sigma_{12} & 0 & \rho_1 \\
    t^5\sigma_{02} & 0 & \rho_2 \\
    t^7\sigma_{01} & 0 & \rho_3 \\
    t^9\sigma_{2} & 0 & \rho_4 \\
    t^{11}\sigma_{1} & 0 & \rho_5 \\
    t^{13}\sigma_{0} & 0 & \rho_6
  \end{array}
\right)
\qquad
\left(
  \begin{array}{c|c|c}
    0 & \sigma_0 & 0 \\
    0 & \sigma_1 & 0 \\
    0 & \sigma_2 & 0 \\
    t^3\sigma_0-t^2\sigma_1 & \sigma_{01} & 0 \\
    t^3\sigma_1-t^2\sigma_2 & -t\sigma_{01}+\sigma_{02} & 0 \\
    0 & t^2\sigma_{01}-t\sigma_{02}+\sigma_{12} & 0 \\
    t^3\sigma_{01}-t^2\sigma_{02}+t\sigma_{12} & \sigma_{012} & 0 \\
    -t\sigma_{012} & 0 & \rho_0 \\
    -t^3\sigma_{12} & 0 & \rho_1 \\
    -t^5\sigma_{02} & 0 & \rho_2 \\
    0 & t^4\sigma_{012} & \rho_3-t\rho_2+t^2\rho_1 \\
    -t^9\sigma_{2} & 0 & \rho_4 \\
    0 & t^7\sigma_{12} & \rho_5-t\rho_4 \\
    0 & t^9\sigma_{02} & \rho_6-t^2\rho_4
  \end{array}
\right)
\]

This computation provides us with a module of generators with 6 elements, that injects into the generators $\sigma_0,\dots,\sigma_{012}$ to hit the elements $\sigma_0,\sigma_1,\sigma_2,t^2\sigma_{01}-t\sigma_{02}+\sigma_{12},t^4\sigma_{012},t^7\sigma_{12},t^9\sigma_{02}$. We shall call the generators of the kernel module $\kappa_0,\dots,\kappa_6$ with the images in the order stated here. They come in degrees $0,1,2,5,10,12,13$.

To compute the relations for the kernel module, we setup another matrix for a Gaussian elimination; matrix on the left, result of the elimination on the right:
\[
\left(
  \begin{array}{c|c|c}
    \sigma_0 & \kappa_0 & 0 \\
    \sigma_1 & \kappa_1 & 0 \\
    \sigma_2 & \kappa_2 & 0 \\
    t^2\sigma_{01}-t\sigma_{02}+\sigma_{12} & \kappa_3 & 0 \\
    -t\sigma_{012} & 0 & \rho_0 \\
    -t^3\sigma_{12} & 0 & \rho_1 \\
    -t^5\sigma_{02} & 0 & \rho_2 \\
    -t^7\sigma_{01} & 0 & \rho_3 \\
    t^4\sigma_{012} & \kappa_4 & 0 \\
    -t^9\sigma_{2} & 0 & \rho_4 \\
    -t^{11}\sigma_{1} & 0 & \rho_5 \\
    t^7\sigma_{12} & \kappa_5 & 0 \\
    -t^{13}\sigma_{0} & 0 & \rho_6 \\
    t^9\sigma_{02} & \kappa_6 & 0 \\
  \end{array}
\right)
\qquad
\left(
  \begin{array}{c|c|c}
    \sigma_0 & \kappa_0 & 0 \\
    \sigma_1 & \kappa_1 & 0 \\
    \sigma_2 & \kappa_2 & 0 \\
    t^2\sigma_{01}-t\sigma_{02}+\sigma_{12} & \kappa_3 & 0 \\
    -t\sigma_{012} & 0 & \rho_0 \\
    -t^3\sigma_{12} & 0 & \rho_1 \\
    -t^5\sigma_{02} & 0 & \rho_2 \\
    0 & t^5\kappa_3 & \rho_3-t\rho_2+t^2\rho_1 \\
    0 & \kappa_4 & t^3\rho_0 \\
    0 & t^9\kappa_2 & \rho_4 \\
    0 & t^{11}\kappa_1 & \rho_5 \\
    0 & \kappa_5 & -t^4\rho_1 \\
    0 & t^{13}\kappa_0 & \rho_6 \\
    0 & \kappa_6 & t^4\rho_2 \\
  \end{array}
\right)
\]

Reading off the kernel from the bottom of this reduced matrix gives us the complete presentation of the kernel module as:
\[
\ker\partial = 
\frac{\kappa_0,\dots,\kappa_6}
{t^5\kappa_3,\kappa_4,t^9\kappa_2,t^{11}\kappa_1,\kappa_5,t^{13}\kappa_0,\kappa_6}
\]

We can observe that by sheer luck, this presentation is already on Smith normal form, so no further work is needed. Furthermore, three of the relations -- $\kappa_4,\kappa_5,\kappa_6$ occur in the reduced matrix without any factor $t$, indicating they belong to barcodes of length 0 -- in fact, these are relations inherent in the generators at hand, but that were already in the torsion part of the module when they first materialize. This computation gives us the kernel module of the boundary map -- the persistent cycle module. To compute homology, we now need to compute the cokernel of this module under the boundary map. First, we need to express the homology boundaries in the cycle basis. This can be done with yet another Gaussian reduction; result on the right:
\[
\left(
  \begin{array}{c|c|c}
    \sigma_0 & \kappa_0 & 0 \\
    \sigma_1 & \kappa_1 & 0 \\
    \sigma_2 & \kappa_2 & 0 \\
    t^2\sigma_{01}-t\sigma_{02}+\sigma_{12} & \kappa_3 & 0 \\
    t^4\sigma_{012} & \kappa_4 & 0 \\
    t^7\sigma_{12} & \kappa_5 & 0 \\
    t^9\sigma_{02} & \kappa_6 & 0 \\
    0 & 0 & \sigma_0  \\
    0 & 0 & \sigma_1  \\
    0 & 0 & \sigma_2  \\
    -(t^3\sigma_0-t^2\sigma_1) & 0 & \sigma_{01} \\
    -(t^4\sigma_0-t^2\sigma_2) & 0 & \sigma_{02} \\
    -(t^4\sigma_1-t^3\sigma_2) & 0 & \sigma_{12} \\
    -(t^3\sigma_{01}-t^2\sigma_{02}+t\sigma_{12}) & 0 & \sigma_{012} \\    
  \end{array}
\right)
\qquad
\left(
  \begin{array}{c|c|c}
    \sigma_0 & \kappa_0 & 0 \\
    \sigma_1 & \kappa_1 & 0 \\
    \sigma_2 & \kappa_2 & 0 \\
    t^2\sigma_{01}-t\sigma_{02}+\sigma_{12} & \kappa_3 & 0 \\
    t^4\sigma_{012} & \kappa_4 & 0 \\
    t^7\sigma_{12} & \kappa_5 & 0 \\
    t^9\sigma_{02} & \kappa_6 & 0 \\
    0 & 0 & \sigma_0  \\
    0 & 0 & \sigma_1  \\
    0 & 0 & \sigma_2  \\
    0 & t^3\kappa_0-t^2\kappa_1 & \sigma_{01} \\
    0 & t^4\kappa_0-t^2\kappa_2 & \sigma_{02} \\
    0 & t^4\kappa_1-t^3\kappa_2 & \sigma_{12} \\
    0 & t\kappa_3 & \sigma_{012} \\    
  \end{array}
\right)
\]

To compute the cokernel, these images -- $t^3\kappa_0-t^2\kappa_1, t^4\kappa_0-t^2\kappa_2, t^4\kappa_1-t^3\kappa_2, t\kappa_3$ are included among the relations, and a Smith normal form computed. The matrix for this Smith normal form is (0 replaced with $\cdot$ for clarity); again the computed Smith normal form is on the right:
\[
\begin{pmatrix}
  t^3 & t^2 & \cdot & \cdot & \cdot & \cdot & \cdot \\
  -t^4 & \cdot & t^2 & \cdot & \cdot & \cdot & \cdot \\
  \cdot & -t^4 & -t^3 & \cdot & \cdot & \cdot & \cdot \\
  \cdot & \cdot & \cdot & t & \cdot & \cdot & \cdot \\
  \cdot & \cdot & \cdot & t^5 & \cdot & \cdot & \cdot \\
  \cdot & \cdot & \cdot & \cdot & 1 & \cdot & \cdot \\
  \cdot & \cdot & t^9 & \cdot & \cdot & \cdot & \cdot \\
  \cdot & t^{11} & \cdot & \cdot & \cdot & \cdot & \cdot \\
  \cdot & \cdot & \cdot & \cdot & \cdot & 1 & \cdot \\
  t^{13} & \cdot & \cdot & \cdot & \cdot & \cdot & \cdot \\
  \cdot & \cdot & \cdot & \cdot & \cdot & \cdot & 1 \\
\end{pmatrix}
\qquad
\begin{pmatrix}
  \cdot & t^2 & \cdot & \cdot & \cdot & \cdot & \cdot \\
  \cdot & \cdot & t^2 & \cdot & \cdot & \cdot & \cdot \\
  \cdot & \cdot & \cdot & \cdot & \cdot & \cdot & \cdot \\
  \cdot & \cdot & \cdot & t & \cdot & \cdot & \cdot \\
  \cdot & \cdot & \cdot & \cdot & \cdot & \cdot & \cdot \\
  \cdot & \cdot & \cdot & \cdot & 1 & \cdot & \cdot \\
  \cdot & \cdot & \cdot & \cdot & \cdot & \cdot & \cdot \\
  \cdot & \cdot & \cdot & \cdot & \cdot & \cdot & \cdot \\
  \cdot & \cdot & \cdot & \cdot & \cdot & 1 & \cdot \\
  t^{13} & \cdot & \cdot & \cdot & \cdot & \cdot & \cdot \\
  \cdot & \cdot & \cdot & \cdot & \cdot & \cdot & 1 \\
\end{pmatrix}
\]

From this normal form we can read off the final barcode for the homology module:
\[
\begin{array}{r|l}
  \text{\bf Dim.} & \text{\bf Interval} \\ \hline
  0 & (0,13) \\
  0 & (1,3) \\
  0 & (2,4) \\
  1 & (5,6) \\
\end{array}
\]

\end{document}